\documentclass[draftclsnofoot, onecolumn, comsoc, 12pt]{IEEEtran} 

\usepackage{graphicx,epsfig}
\usepackage[noadjust]{cite}
\usepackage{mcite}
\usepackage{fancyhdr}
\usepackage{threeparttable}
\usepackage{epsf,epsfig}
\usepackage{amsthm}
\usepackage{amsmath}
\usepackage{siunitx}
\usepackage{amssymb}

\usepackage{dsfont}
\usepackage{subfigure}
\usepackage{color}
\usepackage[linesnumbered,ruled, noend]{algorithm2e}
\usepackage{algpseudocode}
\usepackage{algcompatible}
\usepackage{enumerate}
\usepackage{gensymb}
\usepackage{cancel}

\newtheorem{lemma}{Lemma}

\newtheorem{remark}{Remark}

\usepackage{bbm}
\usepackage{eucal}

\def\bb{{\bf b}}

\def\be{{\bf e}}
\def\bff{{\bf f}}

\def\bh{{\bf h}}

\def\bn{{\bf n}}

\def\bq{{\bf q}}

\def\bs{{\bf s}}

\def\bv{{\bf v}}
\def\bw{{\bf w}}
\def\bx{{\bf x}}
\def\by{{\bf y}}

\def\bA{{\bf A}}
\def\bB{{\bf B}}
\def\bC{{\bf C}}
\def\bD{{\bf D}}
\def\bE{{\bf E}}

\def\bG{{\bf G}}
\def\bH{{\bf H}}
\def\bI{{\bf I}}

\def\bR{{\bf R}}

\def\bV{{\bf V}}
\def\bW{{\bf W}}

\def\cC{\mbox{$\mathcal{C}$}}

\def\cN{\mbox{$\mathcal{N}$}}
\def\cO{\mbox{$\mathcal{O}$}}

\def\cQ{\mbox{$\mathcal{Q}$}}

\def\bbC{\mbox{$\mathbb{C}$}}

\def\bbE{\mbox{$\mathbb{E}$}}

\begin{document}

\title{\huge Energy Efficiency Maximization Precoding for Quantized Massive MIMO Systems}


\author{Jinseok~Choi, Jeonghun~Park, and Namyoon~Lee

\thanks{J. Choi is with the Department of Electrical Engineering, Ulsan National Institute of Science and Technology (UNIST), Ulsan, 44919, South Korea (e-mail: {\texttt{jinseokchoi@unist.ac.kr}}). 

J. Park is with the School of Electronics Engineering, College of IT Engineering, Kyungpook National University, Daegu, 41566, South Korea (e-mail: {\texttt{jeonghun.park@knu.ac.kr}}).

N. Lee is with the Department of Electrical Engineering, POSTECH, Pohang, Gyeongbuk 37673, South Korea  (e-mail: {\texttt{nylee@postech.ac.kr}}).
}

\thanks{This work was supported in part by the National Research Foundation of Korea (NRF) grants funded by the Korea government (MSIT) (No. 2019R1G1A1094703) and (No. 2021R1C1C1004438), and in part by the MSIT (Ministry of Science and ICT), Korea, under the ITRC (Information Technology Research Center) support program (IITP-2021-2017-0-01635) supervised by the IITP (Institute for Information \& Communications Technology Planning \& Evaluation).}
}

\maketitle \setcounter{page}{1}

\begin{abstract} 
The use of low-resolution digital-to-analog and analog-to-digital converters (DACs and ADCs) significantly benefits energy efficiency (EE) at the cost of high quantization noise in implementing massive multiple-input multiple-output (MIMO) systems. 
For maximizing EE in quantized downlink massive MIMO systems, this paper formulates a precoding optimization problem {\color{black} with antenna selection}; 
yet acquiring the optimal joint precoding {\color{black} and antenna selection} solution is challenging due to the intricate EE characterization. 
To resolve this challenge, we decompose the problem into precoding direction and power optimization problems. 
For precoding direction, we characterize the first-order optimality condition, which entails the effects of quantization distortion {\color{black} and antenna selection}. 
For precoding power, we obtain the optimum solution using a gradient descent algorithm to maximize EE for given precoding direction. 
We cast the derived condition as a functional eigenvalue problem, wherein finding the principal eigenvector attains the best local optimal point. To this end, we propose generalized power iteration based algorithm. 
Alternating these two methods, our algorithm identifies a joint solution of the active antenna set and the precoding direction and power. 
In simulations, the proposed methods provide considerable performance gains. Our results suggest that a few-bit DACs are sufficient for achieving high EE in massive MIMO systems.

\end{abstract}

\section{Introduction}

Massive multiple-input multiple-output (MIMO)  \cite{marzetta:twc:10} is a key enabler for future cellular systems because it can provide substantial gains in both spectral efficiency (SE) and coverage by employing large-scale antenna arrays at a base station (BS). In principle, massive MIMO increases the SE gain by scaling antennas at the BS under ideal conditions. Unfortunately, in practice, the use of very large-antenna elements makes the BS hardware and radio frequency (RF) circuit architectures intricate and also gives rise to excessive energy consumption in the BS. The major sources of excessive energy consumption in massive MIMO implementation are high-resolution quantizers, including digital-to-analog converters (DACs) and analog-to-digital converters (ADCs) \cite{zhang:commmag:18}. As a result, implementing massive MIMO with low-resolution DACs and ADCs are rapidly gaining momentum because of its potential for energy-efficient massive MIMO system design \cite{choi:commmag:20, jacobsson:tcom:17, jacobsson:twc:19}.

The use of the low-resolution quantizers in transmission and reception causes severe quantization error. For example, in the downlink MIMO transmission using low-resolution DACs, the transmitted signals are distorted by the low-resolution DACs, and the quantization error incurs a significant amount of inter-user interference. This fundamentally limits the SE gains in massive MIMO. As a result, it is crucial to incorporate the quantization effects in designing a downlink transmission strategy to maximize communication performance such as SE and energy efficiency (EE) in the massive MIMO using low-resolution quantizers. In particular, design for energy efficient communications is critical in realizing massive MIMO systems. 
However, finding a precoding solution for maximizing the EE under low-resolution quantizers constraints is highly challenging. The challenge involves the non-convexity of the EE function, defined as the sum SE normalized by the total transmission power. Further, since the total transmission power is a function of an active antenna set, it entails a non-smooth function. This non-smooth part makes the optimization problem more difficult to solve.  
In this paper, we make a progress toward finding a sub-optimal solution that jointly identifies a set of active antenna elements and the corresponding linear precoding vectors to maximize EE in downlink quantized massive MIMO systems.
\subsection{Prior Work}


In the literature, linear precoding methods for maximizing EE have been widely studied in high-resolution quantizers setups.
For instance, the EE of a massive MIMO uplink system was analyzed in \cite{ngo:tcom:13} by adopting traditional linear receive beamforming such as zero-forcing (ZF) and minimum mean-square-error (MMSE), and it was shown that each user's transmit power can be reduced as the number of BS antennas increase. Enhancing this result, \cite{bjornson:tit:14} characterized the capacity limit in massive MIMO systems by incorporating non-ideal hardware impacts. 
In \cite{xu:twc:13}, an alternating optimization algorithm that determines an active antenna set and transmit precoding in an iterative water-filling fashion was proposed to maximize the EE.
A comprehensive survey on energy-efficient design is found in \cite{prasad:wirelessmag:17}.  


Despite the abundant previous studies, the aforementioned prior work is not applicable when low-resolution quantizers are used. This is because low-resolution quantizers induce non-negligible non-linear quantization distortion, which may render the performance characterization totally different from a conventional high-resolution system. Motivated by this, there exist several prior works that performed performance analysis for low-resolution quantizers. In \cite{singh:tcom:09}, a single-input single-output (SISO) system with  $1$-bit ADCs was considered, and it was shown that binary phase-shift-keying (BPSK) is a capacity-achieving modulation in a real-valued non-fading SISO channel. 
Considering a fading SISO channel with a $1$-bit quantizer, \cite{mezghani:isit:08} showed that quadrature phase-shift-keying (QPSK) is optimal when the signal-to-noise ratio (SNR) is larger than a certain threshold, and on-off QPSK is optimal otherwise. 
In \cite{mo:tsp:15}, assuming that the BS has perfect channel state information (CSI), the capacity of the MIMO channel with $1$-bit quantizer was characterized. 
In \cite{jacobsson:twc:17}, assuming uplink MIMO with low-resolution ADCs, a closed-form approximation on the SE with Gaussian inputs was derived by using Bussgang theorem \cite{bussgang:tech:52}. 
In \cite{choi2019two}, the optimality conditions of an analog combining were derived for uplink massive MIMO systems with low-resolution ADCs, and a near-optimal analog combiner was further developed based on the conditions.

More relevant to this work, several prior works studied precoding methods in downlink systems with low-resolution DACs. 
A main obstacle in developing precoders is that non-linear quantization distortion is not tractable. 
To handle the non-linear quantization distortion in a tractable way, 
in \cite{mezghani:icecs:09}, the non-linearity of the quantization distortion was resolved by adopting Bussgang theorem \cite{bussgang:tech:52}, then by using the linearization, a closed-form expression of an MMSE precoder was derived. 
In addition, a linear precoder was also designed based on the MMSE incorporating the quantization distortion in \cite{mezghani:icecs:09}. 
Extending the result, \cite{usman:icassp:16} developed a gradient projection-based precoding algorithm for downlink massive MIMO systems with $1$-bit DACs. In \cite{usman:icc:17}, by applying the similar approach to \cite{usman:icassp:16}, the joint design of the precoder and the receiver was studied.
In \cite{jacobsson:tcom:17}, a linear precoder for $3$ to $4$-bit DACs and a non-linear precoder for $1$-bit DACs were proposed. Especially, in designing the linear precoder, it adopted a conventional linear precoder such as MMSE or ZF, and quantized the adopted linear precoder to use with low-resolution DACs.
A key finding of \cite{jacobsson:tcom:17} is that using $3$ to $4$-bit DACs offers comparable performance with high-resolution DACs, provided that proper design of precoding is used. This finding is supported by the SE and bit-error-rate (BER) analyses in \cite{jacobsson:twc:19}. 
A similar approach was also used in \cite{saxena:tsp:17}, where conventional ZF precoding is applied, thereafter $1$-bit quantization is conducted to the precoded signal. 
The approach used in \cite{jacobsson:tcom:17} was further improved by using alternating minimization in \cite{chen:tvt:18}. 
As shown in the prior work, Bussgang decomposition is useful to characterize a $1$-bit quantizer, while its tractability is highly limited if a quantizer is more than $2$ bits. 
For precoding design in more general bits quantizers, an additive quantization noise model (AQNM) was also used in several works \cite{ribeiro:jstsp:18, dai:access:19,choi2019base,vlachos:twc:21}, providing a linear approximation of a quantized signal.
Employing the AQNM, \cite{ribeiro:jstsp:18} considered hybrid precoding architecture with low-resolution DACss in a point-to-point MIMO channel and performed extensive performance evaluation by using a conventional precoder. 
In \cite{dai:access:19}, a full-duplex MIMO system with low-resolution quantizers was considered and the achievable rate was analyzed by using conventional maximum ratio transmission (MRT). Similar to \cite{ribeiro:jstsp:18}, the AQNM was also used in \cite{dai:access:19}. 
In \cite{choi2019base, vlachos:twc:21}, an algorithm that selects an active antenna set was proposed for a given precoder. 
Without linear modeling, the alternating direction method of multipliers (ADMM) was used to solve an inter-user interference minimization problem in \cite{wang:twc:18}. Based on this result, a general precoder was designed. 
A $1$-bit massive MIMO precoding method that exploits interference in a constructive way was developed in \cite{li:twc:20}.
In addition to precoding design, QAM constellation range design was jointly performed for massive MIMO systems with $1$-bit DAC systems \cite{sohrabi:jstsp:18}.

As explained above, abounding studies provided crucial insights on low-resolution quantizer systems. 
The existing precoding methods, however, are mostly limited to variants of traditional linear precoding methods such as ZF or MMSE. 
Specifically, the traditional linear precoders are firstly adopted, then quantization distortion effects are reflected. Nevertheless, as we will show in this paper, this approach is far from the optimum due to the inherent limitations on the traditional precoders. 
Accordingly, it is necessary to develop a novel precoding method for massive MIMO with low-resolution quantizers. 

\subsection{Contributions}

We investigate a EE maximization problem with regard to precoders in a downlink multiuser massive MIMO system, where low-resolution DACs and low-resolution ADCs are employed at the BS and users, respectively. Our main contributions are summarized as follows:

\begin{itemize}
    \item     We first put forth a precoding strategy for maximizing EE in quantized massive MIMO systems by reformulating the problem. To accomplish the precoding optimization, by adopting the AQNM \cite{fletcher2007robust}, which is a linear approximation technique of the non-linear quantizer function with additional quantization noise, we define the EE of the quantized massive MIMO system. 
    The EE function is the sum SE function normalized by the total transmit power, a fractional programming form. Therefore, it is critical to find an optimal transmit power level while maximizing the SE. 
    Unfortunately, the EE maximization problem in the downlink quantized massive MIMO system is NP-hard, similar to the case using infinite-resolution DACs and ADCs. Therefore, finding a global optimum solution is infeasible using polynomial-time complexity algorithms.  
     To overcome this challenge, we take the Dinkelbach method~\cite{dinkelbach1967nonlinear} to relax a fractional programming, then we decompose the optimization variables, i.e., the precoding vectors, into two parts: 1) scaling and 2) directional components. 
     
     \item  Leveraging this decomposition, we propose an alternating optimization framework for the EE maximization called Q-GPI-EEM.
      To be specific, for a fixed Dinkelbach variable and a directional component, we obtain the optimum power scaling component by using a gradient descent method. Since this sub-optimization problem is convex, gradient descent is sufficient to find the optimum point.
       Subsequently, with the obtained power scaling component, we find the directional component; we derive the first-order optimality condition for the non-convex precoding direction optimization problem. The derived condition is cast as a functional eigenvalue problem, and it insinuates that a local optimal point that has zero gradient is obtained by finding the principal eigenvector of the functional eigenvalue problem. To this end, modifying the algorithm in \cite{choi:twc:20}, we present an precoding algorithm called quantized generalized power iteration-based algorithm for direction optimization (Q-GPI-DO) that iteratively identifies the principal eigenvector with a few number of iterations. 
       The algorithm iterates by updating the Dinkelbach variable until it converges. Unlike other existing algorithms for the EE maximization in the quantized massive MIMO systems, the most prominent feature of Q-GPI-EEM is to jointly identify a set of active antennas and the corresponding precoding direction by considering the effect of RF circuit power consumption for the active transmit antennas. 
    
    \item  As a byproduct, we also present a precoding method for maximizing the SE in quantized massive MIMO systems, called Q-GPI-SEM by showing that the SE maximization is a special case of the EE maximization.
    Q-GPI-SEM ensures to find a local optimal point of the SE maximization for any important system parameters, including the number of antennas, the number of downlink users, and the number of DAC and ADS resolution bits. Besides, our algorithm generalizes and the prior GPI-based algorithm \cite{choi:twc:20} by incorporating the quantization error effects caused by DACs and ADCs.

    \item Simulation results demonstrate that the proposed algorithms, Q-GPI-EEM and Q-GPI-SEM considerably outperform conventional algorithms in both EE and SE for all system configurations. 
    In terms of EE, our EE maximization algorithm provides robustness to the maximum available transmit power constraint. For example, the EE performances of other methods degrade as the transmit power increases since the SE does not efficiently improve as the transmit power increases in a high signal-to-noise-ratio (SNR) regime. For this reason, using conventional methods, e.g., \cite{vlachos:twc:21}, the EE converges to zero eventually if the transmit power continues to increase. On the contrary, our algorithm offers the monotonically increasing EE by adjusting the actual transmit power in the algorithm. 
    Further, in the EE maximization, antenna selection shows a trend that high-resolution DAC equipped antennas are less likely to be selected for saving the power consumption. 
     In terms of SE, the SE is saturated as the transmit power increases because of the quantization distortion effects of low-resolution quantizers. Using the proposed algorithm, we can pull up this saturation level more than $2\times$  compared to the conventional methods.

    \item In addition to the SE and EE improvement, we also elucidate a system design guideline for quantized massive MIMO systems. Our key finding is that using large-scale antenna elements, each with low-resolution quantizers, is beneficial in both the SE and EE. 
    More detailed system design insights are provided as follows:
    $(i)$ Regarding the EE, there exists the optimal number of the DAC bits that maximize the EE. As the number of antennas increases, the optimal DAC bits decreases, and the EE increases. For instance, using $4$-bit DACs achieves the maximum EE if the BS has $8$ antennas while using $3$-bit DACs achieves the maximum EE if the BS has $32$ antennas. 
     $(ii)$ Regarding the SE, our method is highly efficient if the number of antennas is enough. In particular, the proposed method achieves a similar level of SE of the high-resolution DACs, e.g., $9\sim11$ bits even with low-resolution DACs, e.g., $3\sim5$ bits. This is because the harmful effects from low-resolution quantizers can be compensated by using large spatial degrees-of-freedom provided by massive MIMO.
    $(iii)$ Under the constraint of the total number of DAC bits, using the homogeneous DACs at the BS is beneficial in terms of the SE and the EE. In addition, the proposed algorithm shows higher robustness for DAC configuration, and thus, it can provide more flexibility in the system design. 
    Overall, the proposed algorithm provides high SE and EE performance in the massive MIMO regime, allowing the BS to employ low-resolution DACs and offering high system design flexibility.
\end{itemize}





{\it Notation}: $\bf{A}$ is a matrix and $\bf{a}$ is a column vector. 
Superscripts $(\cdot)^{*}$, $(\cdot)^{\sf T}$, $(\cdot)^{\sf H}$, and $(\cdot)^{-1}$ denote conjugate, transpose, Hermitian, and matrix inversion, respectively. ${\bf{I}}_N$ is an identity matrix of size $N \times N$ and ${\bf{0}}_{M\times N}$ is a zero matrix of size $M \times N$.
$\mathcal{CN}(\mu, \sigma^2)$ is a complex Gaussian distribution with mean $\mu$ and variance $\sigma^2$.  
${\rm Unif}[a, b]$ denotes a discrete uniform distribution from $a$ to $b$.
A diagonal matrix $\rm diag(\bf A)$ has the diagonal entries of $\bA$ at its diagonal entries.
Assuming that ${\bf{A}}_1, \dots, {\bf{A}}_N \in \mathbb{C}^{K \times K}$, ${\rm blkdiag}\left({\bf{A}}_1,\dots, {\bf{A}}_N \right)$ is a block diagonal matrix of size ${KN\times KN}$ whose $n$th block diagonal entry is $\bA_n$. $\|\bf A\|$ represents L2 norm, $\|{\bf A}\|_{F}$ represents Frobenius norm, $\mathbb{E}[\cdot]$ represents an expectation operator, ${\rm tr}(\cdot)$ denotes a trace operator, ${\rm vec}(\cdot)$ indicates a vectorization operator, and $\otimes$ is Kroncker product.


\section{System Model} \label{sec:sys_model}

We consider a downlink multiuser massive MIMO system where the base station (BS) is equipped with $N\gg 1$ antennas and each user is equipped with a single antenna, and there are $K$ users to be served. 
We further assume that the BS employs low-resolution DACs and the users employ low-resolution ADCs. Throughout this paper, we consider a general case for low-resolution quantizers where each DAC and ADC can have any bit configuration.
At the BS, a precoded digital baseband transmit signal vector $\bx \in \bbC^{K}$ is expressed as
\begin{align} \label{eq:digital_signal}
    \bx = \sqrt{P}\bW\bs,
\end{align}
where $\bs\sim \cC\cN({\bf 0}_{K\times 1}, \bI_K)$ is a symbol vector, $\bW \in \bbC^{N\times K}$ represents a precoding matrix, and $P$ is the maximum transmit power. We assume $\| {\bf{W}} \|_{F}^2 \le 1$, thereby the maximum transmit power is maintained as $P$. 

The digital baseband signal $\bx$ in \eqref{eq:digital_signal} is further quantized at the DACs prior to transmission. To characterize the quantized signal, we adopt the AQNM \cite{fletcher2007robust}, 
that approximates the quantization process in a linear form.
Then, the analog baseband transmit signal after quantization becomes
\begin{align}
    \cQ(\bx) \approx \bx_{\sf q} =\sqrt{P} {\pmb \Phi}_{\alpha_{\sf bs}}\bW\bs + \bq_{\sf bs},
\end{align}
where $\mathcal{Q}(\cdot)$ is an element-wise quantizer that applies for each real and imaginary part, ${\pmb \Phi}_{\alpha_{\sf bs}} = {\rm diag}(\alpha_{{\sf bs}, 1}, \dots, \alpha_{{\sf bs}, N}) \in \bbC^{N\times N}$ denotes a diagonal matrix of quantization loss, and $\bq_{\sf bs} \in \bbC^{N}$ is a quantization noise vector. The quantization loss of the $n$th DAC  $\alpha_{{\sf bs}, n} \in (0,1)$
is defined as $\alpha_{{\sf bs},n} = 1- \beta_{{\sf bs},n}$, where $\beta_{{\sf bs},n}$ is a normalized mean squared quantization error with $\beta_{{\sf bs},n} = \frac{\mathbb{E}[|{x} - \mathcal{Q}_n(x)|^2]}{\mathbb{E}[|{x}|^2]}$ 
\cite{fletcher2007robust, zhang2018mixed}.
The values of $\beta_{{\sf bs},n}$ are characterized depending on the number of quantization bits at the $n$th BS antenna $b_{{\sf DAC},n}$. Specifically, for $b_{{\sf DAC},n} \leq 5$, $\beta_{{\sf bs},n}$ is quantified in Table 1 in \cite{fan2015uplink}. For $b_{{\sf DAC},n} > 5$, $\beta_{{\sf bs},n}$ can be approximated as $\beta_{{\sf bs},n} \approx \frac{\pi\sqrt{3}}{2}2^{-2b_{{\sf DAC},n}}$ \cite{gersho2012vector}.
The quantization noise $\bq_{\sf bs}$ is uncorrelated with $\bx$ and follows $\bq_{\sf bs} \sim \mathcal{CN}({\bf 0}_{N\times 1}, \mathbf{R}_{\bq_{\sf bs}\bq_{\sf bs}})$. The covariance $\bR_{\bq_{\sf bs}\bq_{\sf bs}}$ is computed as \cite{fletcher2007robust,zhang2018mixed} 
\begin{align}
    \bR_{\bq_{\sf bs}\bq_{\sf bs}} &= {\pmb \Phi}_{\alpha_{\sf bs}}{\pmb \Phi}_{\beta_{\sf bs}} {\rm diag}\left(\bbE\big[\bx\bx^{\sf H}\big]\right) \\
    &= {\pmb \Phi}_{\alpha_{\sf bs}}{\pmb \Phi}_{\beta_{\sf bs}} {\rm diag}\left(P\bW\bW^{\sf H}\right).
\end{align}
The quantized signal $\bx_{\sf q}$ is amplified by a power amplifier under a power constraint at the BS. 
Since the maximum transmit power is defined as $P$, we have the following power constraint~\cite{choi:tcom:comp} 
\begin{align}
    {\rm tr}\left(\bbE\big[\bx_{\sf q}\bx_{\sf q}^{\sf H}\big]\right) \leq P.
\end{align}

Now, we represent the received analog baseband signals at all the $K$ users as
\begin{align}
    \label{eq:rx_analog}
    \by = \bH^{\sf H}\bx_{\sf q} + \bn,
\end{align}
where $\bH \in \bbC^{N\times K}$ is a channel matrix and $\bn \sim \cC\cN({\bf 0}_{K\times 1}, \sigma^2\bI_K)$ indicates an $K\times 1$ additive white Gaussian noise vector with zero mean and  variance of $\sigma^2$. 
Each column of the channel matrix $\bh_k$ represents the channel between user $k$ and the BS.
We assume a block fading model, where ${\bf{h}}_k$ is invariant within one transmission block and changes independently over  transmission blocks. 
We focus on a single transmission block and assume that the CSI is perfectly known at the BS. 

The received analog signals in \eqref{eq:rx_analog} are quantized at the ADCs of the users. 
Accordingly, the received digital baseband signals are given as \cite{fletcher2007robust, choi:tsp:18}
\begin{align}
    \cQ(\by) \approx \by_{\sf q} &= {\pmb \Phi}_{\alpha}\by + \bq \\
    & = {\pmb \Phi}_{\alpha}\bH^{\sf H}\bx_{\sf q} + {\pmb \Phi}_{\alpha}\bn + \bq \\
    & =  \sqrt{P}{\pmb \Phi}_{\alpha}\bH^{\sf H}{\pmb \Phi}_{\alpha_{\sf bs}}\bW\bs  + {\pmb \Phi}_{\alpha}\bH^{\sf H}\bq_{\sf bs}+ {\pmb \Phi}_{\alpha}\bn +  \bq,
\end{align}
where $\pmb \Phi_\alpha = {\rm diag}(\alpha_1,\dots,\alpha_K) \in \bbC^{K\times K}$ is a diagonal matrix of ADC quantization loss defined as $\alpha_k = 1-\beta_k$, and $\bq \in \bbC^{K}$ is a quantization noise vector, which is uncorrelated with $\by$.
Here, $\alpha_k$ and $\beta_k$ are similarly defined as $\alpha_{{\sf bs},n}$ and  $\beta_{{\sf bs},n}$, respectively. 
The quantization noise $\bq$ has zero mean and follows a complex Gaussian distribution.
Consequently, the digital baseband signal at user $k$ is given as
\begin{align}
    \label{eq:yqk}
    y_{{\sf q}, k} =\sqrt{P}\alpha_k  \bh^{\sf H}_k {\pmb \Phi}_{\alpha_{\sf bs}}\bw_k s_k +  \sqrt{P} \sum_{\ell \neq k}\alpha_k \bh^{\sf H}_k {\pmb \Phi}_{\alpha_{\sf bs}}\bw_\ell s_\ell + \alpha_k\bh_k^{\sf H}\bq_{\sf bs} + \alpha_k n_k + q_k,
\end{align}
where $\bw_k \in \bbC^{N}$  denotes the $k$th column vector of $\bW$, and  $s_k$, $n_k$, and $q_k$ represent $k$th element of $\bs$, $\bn$, and  $\bq$, respectively.
The variance of $q_k$ is computed as \cite{fletcher2007robust}
\begin{align}
    r_{{\sf q}_k{\sf q}_k}& =\alpha_k\beta_k\bbE\big[y_k y_k^{\sf H}\big] \\
    & = \alpha_k\beta_k\left(\bh_k^{\sf H}\bbE\big[\bx_{\sf  q}\bx_{\sf q}^{\sf H}\big]\bh_k^{\sf H}+ \sigma^2\right)\\
    \label{eq:rqq}
    &  = \alpha_k\beta_k\left(\bh_k^{\sf H}\big(P{\pmb \Phi}_{\alpha_{\sf bs}}\bW\bW^{\sf H}{\pmb \Phi}_{\alpha_{\sf bs}}^{\sf H} + \bR_{\bq_{\sf  bs}\bq_{\sf  bs}}\big)\bh_k + \sigma^2 \right).
\end{align}

In the following sections, we formulate an EE maximization problem, propose an algorithm, and validate the performance via simulations accordingly.


\section{Energy Efficiency Maximization Problem Formulation} 

\subsection{Performance Metric}

We characterize the EE for the considered system as our main performance metric.
Using \eqref{eq:yqk}, we define the downlink SINR of user $k$  as
\begin{align}
    \label{eq:sinr}
    \Gamma_k = \frac{P\alpha_k^2    |\bh^{\sf H}_k {\pmb \Phi}_{\alpha_{\sf bs}}\bw_k|^2}  { P\alpha_k^2 \sum_{\ell \neq k}  |\bh^{\sf H}_k {\pmb \Phi}_{\alpha_{\sf bs}}\bw_\ell|^2 + \alpha_k^2\bh_k^{\sf H}\bR_{\bq_{\sf bs}\bq_{\sf bs}}\bh_k + \alpha_k^2\sigma^2 + r_{{\sf q}_k{\sf q}_k}}.
\end{align}
Unlike a perfect quantization system where resolution of DACs and ADCs is infinite, the SINR in \eqref{eq:sinr} includes the quantization noise power and the quantization loss induced by both low-resolution DACs and ADCs.
Accordingly, the SE for user $k$ is expressed as
\begin{align}
    \label{eq:se}
    R_{k} = \log_2 \left(1 + \Gamma_k\right). 
\end{align}

Now, we define the EE based on \eqref{eq:se}.
Let $P_{\sf LP}$, $P_{\sf M}$, $P_{\sf LO}$, $P_{\sf H}$, $P_{\sf PA}$, $P_{\sf RF}$, $P_{\sf DAC}$, $P_{\sf cir}$, and $P_{\sf BS}$ denote the power consumption of low-pass filter, mixer,
local oscillator, $90^{\circ}$ hybrid with buffer, power amplifier (PA), radio frequency (RF) chain,  DAC, analog circuits, and BS, respectively.
The DAC power consumption  $P_{\sf DAC}$ (in Watt) is defined as 
\cite{ribeiro:jstsp:18, cui2005energy}
\begin{align}
    P_{\sf DAC}(b_{\sf DAC},f_s) = 1.5\times10^{-5}\cdot 2^{b_{\sf DAC}} + 9\times 10^{-12}\cdot f_s \cdot b_{\sf DAC}.
\end{align}
If antenna $n$  is active, then the corresponding DAC pair and  RF chain  consume the circuit power of $2\,P_{\sf DAC}(b_{{\sf DAC},n},f_s) + P_{\sf RF}$ where  $P_{\sf RF} = 2P_{\sf LP} + 2 P_{\sf M} + P_{H}$. On the contrary, if the antenna is inactive, no power is consumed in the corresponding DAC and RF chain. 
Considering such behaviour, the circuit power consumption $P_{\sf cir}$ is formulated as \cite{ribeiro:jstsp:18}
\begin{align} 
    \label{eq:P_cir}
    P_{\sf cir}= P_{\sf LO}  + \sum_{n = 1}^{N} \mathbbm{1}_{\{n \in \CMcal{A}\}} \big(2\,P_{\sf DAC}(b_{{\sf DAC},n},f_s) + P_{\sf RF}\big),
\end{align}
where $\mathbbm{1}_{\{a\}}$ is the indicator function which is $\mathbbm{1}_{\{a\}} = 1 $ only when $a$ is true and $\mathbbm{1}_{\{a\}}= 0$ otherwise, and $\CMcal{A}$ is an index set of active antennas.
Then the BS power consumption $P_{\sf BS}$ is given as \cite{ribeiro:jstsp:18} 
\begin{align}
    P_{\sf BS}(N, \bb_{\sf DAC}, f_s, \bW) 
    = P_{\sf cir} + \kappa^{-1}P_{\sf tx},
\end{align}
where $P_{\sf tx}$ is the actual transmit power, $\kappa$ is a PA efficiency, i.e., $\kappa = P_{\sf tx}/P_{\sf PA}$, and  $\bb_{\sf DAC} = [b_{{\sf DAC},1},\dots,b_{{\sf DAC},N}]^{\sf T}$.
Finally, the EE of the considered system is defined as
\begin{align}
    \label{eq:ee}
    \eta = \frac{{\Omega}\sum_{k=1}^K R_k}{P_{\sf BS}(N, \bb_{\sf DAC}, f_s, \bW)}, 
\end{align}
where $\Omega$ denotes transmission bandwidth.




\subsection{Formulated Problem}

We now formulate an EE maximization problem with respect to a precoder.  
We first normalize the EE $\eta$ in \eqref{eq:ee} by the bandwidth $\Omega$ since it is irrelevant to precoder design. Throughout the paper, we use the EE and normalized EE interchangeably as they do not change the problem.
Then the EE maximization problem is formulated as
\begin{align} 
    \label{eq:problem_ee}
    \mathop{{\text{maximize}}}_{\bW}& \;\; \frac{\sum_{k=1}^K R_k}{P_{\sf BS}(N, \bb_{\sf DAC}, f_s, \bW)} \\
    \label{eq:const_ee}
    {\text{subject to}} & \;\;  {\rm tr}\left(\bbE\big[\bx_{\sf q}\bx_{\sf q}^{\sf H}\big]\right) \leq P.
\end{align}
We remark that the BS power consumption includes the circuit power consumption, which is a function of active antennas, and thus, the EE maximization problem needs to be solved by designing $\bW$ with incorporating the impact of active and inactive antenna sets. 
In this regard, we propose an algorithm that jointly designs a precoder and performs an antenna selection in Section~\ref{sec:algorithm}.

\section{Joint Precoding and Antenna Selection Algorithm} \label{sec:algorithm}

A direct solution for the formulated problem in \eqref{eq:problem_ee} is not available since it is non-smooth and non-convex. 
To resolve these challenges, we first reformulate the problem and then propose an algorithm that provides the best sub-optimal solution.



\subsection{Problem Reformulation}

We  reformulate \eqref{eq:problem_ee}  by using the Dinkelbach method~\cite{dinkelbach1967nonlinear} as
\begin{align} 
    \label{eq:problem_ee_dinkelbach}
    \mathop{{\text{maximize}}}_{\bW, \mu}& \;\; \sum_{k=1}^K R_k - \mu {P_{\sf BS}(N, \bb_{\sf DAC}, f_s, \bW)} \\
    \label{eq:const_ee_dinkelbach}
    {\text{subject to}} & \;\;  {\rm tr}\left(\bbE\big[\bx_{\sf q}\bx_{\sf q}^{\sf H}\big]\right) \leq P\\
    & \;\;\mu > 0,
\end{align}
where $\mu$ is an auxiliary variable.
To cast the problem \eqref{eq:problem_ee_dinkelbach} into a tractable form, 
we first rewrite the DAC quantization noise covariance term coupled with a user channel in  \eqref{eq:sinr}  as
\begin{align}
    \bh_k^{\sf H}\bR_{\bq_{\sf bs}\bq_{\sf bs}}\bh_k &= \bh_k^{\sf H}{\pmb \Phi}_{\alpha_{\sf bs}}{\pmb \Phi}_{\beta_{\sf bs}}{\rm diag}\left(P\bW\bW^{\sf H}\right)\bh_k \\
    &= P\sum_{n=1}^N |h_{n,k}|^2\alpha_{{\sf bs},n}\beta_{{\sf bs},n}\sum_{\ell=1}^K|w_{n,\ell}|^2\\
    & =P\sum_{\ell =1}^K\sum_{n=1}^N w_{n,\ell}^*\alpha_{{\sf bs},n}\beta_{{\sf bs},n}h_{n,k} h_{n,k}^*w_{n,\ell} \\
    \label{eq:Rqq_reform}
    & = P\sum_{\ell=1}^K\bw_\ell^{\sf H}{\pmb \Phi}_{\alpha_{\sf bs}}{\pmb \Phi}_{\beta_{\sf bs}}{\rm diag}\left(\bh_k\bh_k^{\sf H}\right)\bw_\ell.
\end{align}
We also rewrite the ADC quantization noise variance $r_{{\sf q}_k{\sf q}_k}$ in \eqref{eq:rqq} as
\begin{align}
    r_{{\sf q}_k{\sf q}_k}  &\stackrel{(a)}=  \alpha_k\beta_k\left(P\bh_k^{\sf H}{\pmb \Phi}_{\alpha_{\sf bs}}\bW\bW^{\sf H}{\pmb \Phi}_{\alpha_{\sf bs}}^{\sf H}\bh_k  +  P\sum_{\ell=1}^K\bw_\ell^{\sf H}{{\pmb \Phi}}_{\alpha_{\sf bs}}{{\pmb \Phi}}_{\beta_{\sf bs}}{\rm diag}\left(\bh_k\bh_k^{\sf H}\right)\bw_\ell + \sigma^2\right)  \\
    & = \alpha_k\beta_k\left(P\sum_{\ell=1}^K\bw_\ell^{\sf H} {\pmb \Phi}_{\alpha_{\sf bs}}^{\sf H}\bh_k\bh_k^{\sf H}{\pmb \Phi}_{\alpha_{\sf bs}}\bw_\ell + P\sum_{\ell=1}^K\bw_\ell^{\sf H}{{\pmb \Phi}}_{\alpha_{\sf bs}}{{\pmb \Phi}}_{\beta_{\sf bs}}{\rm diag}\left(\bh_k\bh_k^{\sf H}\right)\bw_\ell+ \sigma^2\right)\\
    \label{eq:rqq_reform}
    & = P\alpha_k\beta_k\sum_{\ell=1}^K\bw_\ell^{\sf H}\left({\pmb \Phi}_{\alpha_{\sf bs}}^{\sf H}\bh_k\bh_k^{\sf H}{\pmb \Phi}_{\alpha_{\sf bs}} + {{\pmb \Phi}}_{\alpha_{\sf bs}}{{\pmb \Phi}}_{\beta_{\sf bs}}{\rm diag}\left(\bh_k\bh_k^{\sf H}\right)\right)\bw_\ell  + \alpha_k\beta_k\sigma^2,
\end{align}
where $(a)$ comes from \eqref{eq:Rqq_reform}.
Applying \eqref{eq:Rqq_reform}, \eqref{eq:rqq_reform}, and $\beta_k = 1-\alpha_k$ to the SINR in \eqref{eq:sinr}, we re-organize the SINR term to represent the SINR in a more compact form:
\begin{align}
    \label{eq:sinr_reform}
    \Gamma_k = \frac{\alpha_k|\bh^{\sf H}_k {\pmb \Phi}_{\alpha_{\sf bs}}\bw_k|^2}
    { \sum_{\ell = 1}^{K}  |\bh^{\sf H}_k{\pmb \Phi}_{\alpha_{\sf bs}} \bw_\ell|^2 - \alpha_k|\bh_k^{\sf H}{\pmb \Phi}_{\alpha_{\sf bs}}\bw_k|^2  +  \sum_{\ell = 1}^{K}\bw_\ell^{\sf H}{\pmb \Phi}_{\alpha_{\sf bs}}{\pmb \Phi}_{\beta_{\sf bs}}{\rm diag}\left(\bh_k\bh_k^{\sf H}\right)\bw_\ell+ \sigma^2/P }.
\end{align}
Now, we simplify the transmit power constraint in \eqref{eq:const_se} as
\begin{align}
    \label{eq:const_se_simple}
    {\rm tr}\left(\bbE\big[\bx_{\sf q}\bx_{\sf q}^{\sf H}\big]\right)= {\rm tr}\left(P{\pmb \Phi}_{\alpha_{\sf bs}}\bW\bW^{\sf H}{\pmb \Phi}_{\alpha_{\sf bs}}^{\sf H}+ P{\pmb \Phi}_{\alpha_{\sf bs}}{\pmb \Phi}_{\beta_{\sf bs}}{\rm diag}\left(\bW\bW^{\sf H}\right)\right) \stackrel{(a)}= {\rm tr}\left(P{\pmb \Phi}_{\alpha_{\sf bs}}\bW\bW^{\sf H}\right) \leq P,
\end{align}
where $(a)$ comes from ${\pmb \Phi}_{\beta_{\sf bs}} = \bI_{N}-{\pmb \Phi}_{\alpha_{\sf bs}}$, ${\rm tr}\big({\pmb \Phi}_{\alpha_{\sf bs}}\bW\bW^{\sf H}{\pmb \Phi}_{\alpha_{\sf bs}}\big) = {\rm tr}\big({\pmb \Phi}_{\alpha_{\sf bs}}^2\bW\bW^{\sf H}\big) $, and the definition of the trace operation.

Regarding the SE point of view, using the maximum transmit power $P$ maximizes the SE for a given precoder. 
For the EE point of view, however, using the maximum transmit power may decrease the EE because it increases the SE with a logarithmic scale while it increases the power consumption with a (approximately) linear scale. For this reason, using the maximum power without adjusting a power level, the EE converges to zero eventually as shown in \cite{vlachos:twc:21}. 
To prevent this phenomenon, an optimal power scaling solution needs to be found to maximize the EE. To this end, we introduce a scalar weight $\tau:0<\tau \leq 1$ in the power constraint as
\begin{align}
    \label{eq:const_ee_reform}
    {\rm tr}\left(\bbE\big[\bx_{\sf q}\bx_{\sf q}^{\sf H}\big]\right) = \tau P.
\end{align}
Then applying \eqref{eq:const_ee_reform} to \eqref{eq:const_se_simple}, the power constraint reduces to
\begin{align}
    \label{eq:const_ee_with_tau}
    {\rm tr}\left({\pmb \Phi}_{\alpha_{\sf bs}}\bW\bW^{\sf H}\right) \leq \tau.
\end{align}

To incorporate the power scaling in the precoder, we define a weighted and normalized precoding matrix $\bV = [\bv_1,\dots,\bv_K]$ where
\begin{align}
    \label{eq:v_def}
        \bv_k = \frac{1}{\sqrt{\tau}}{\pmb \Phi}_{\alpha_{\sf bs}}^{1/2}\bw_k  \in \bbC^{N}.
\end{align}
Then, applying \eqref{eq:v_def} to \eqref{eq:sinr_reform}, the SE of user $k$ is re-written as a function of $\tau$ and $\bV$, which is
\begin{align}
    \label{eq:ee_Rk}
   &\bar{R}_k(\bV,\tau) 
   \\    \nonumber 
   &=\log_2\left(1+\frac{\tau\alpha_k|\bh^{\sf H}_k {\pmb \Phi}_{\alpha_{\sf bs}}^{1/2}\bv_k|^2}
    { \tau\sum_{\ell = 1}^{K}  |\bh^{\sf H}_k{\pmb \Phi}_{\alpha_{\sf bs}}^{1/2} \bv_\ell|^2 -\tau \alpha_k|\bh_k^{\sf H}{\pmb \Phi}_{\alpha_{\sf bs}}^{1/2}\bv_k|^2  +  \tau\sum_{\ell = 1}^{K}\bv_\ell^{\sf H}{\pmb \Phi}_{\beta_{\sf bs}}{\rm diag}\left(\bh_k\bh_k^{\sf H}\right)\bv_\ell+ \sigma^2/P }\right).
\end{align}
Using \eqref{eq:const_ee_with_tau}, \eqref{eq:v_def}, and \eqref{eq:ee_Rk}, the problem in \eqref{eq:problem_ee_dinkelbach} becomes
\begin{align} 
    \label{eq:problem_ee_dinkelbach2}
    \mathop{{\text{maximize}}}_{\bV, \tau, \mu}& \;\; \sum_{k=1}^K \bar{R}_k(\bV, \tau) - \mu {P_{\sf BS}(N, \bb_{\sf DAC}, f_s, \bV, \tau)} \\
    \label{eq:const_ee_dinkelbach2}
    {\text{subject to}} & \;\;  {\rm tr}\left(\bV\bV^{\sf H}\right) = 1
    \\
    & \;\;  0<\tau \leq 1
    \\
    & \;\;  \mu >0.
\end{align}
Note that the power constraint is equivalent to $\|\bar \bv\|^2=1$ where $\bar \bv = {\rm vec}(\bV)$, which is interpreted as a directional component of the precoder.
Consequently, the EE maximization problem is now a problem of jointly finding an optimal power level $\tau$ and precoding direction $\bar \bv$.
In this regard, we decompose \eqref{eq:problem_ee_dinkelbach2} into two subsequent problems: find $\bV$ for given $\tau$ and find $\tau$ for a given $\bV$. With the obtained $\tau$ and $\bV$, we iteratively update the Dinkelbach variable $\mu$. 


\subsection{Proposed Algorithm}

\subsubsection{Optimal Direction $\bV^\star$}

we solve the problem in \eqref{eq:problem_ee_dinkelbach2} regarding $\bV$ for given $\tau$ and $\mu$.
Again,  this phase is mainly related to designing an optimal direction of precoding since the power constraint becomes $\|\bar \bv\|^2 = 1$.
Let  $\bG_k = ({\pmb \Phi}_{\alpha_{\sf bs}}^{1/2})^{\sf H} {\bf{h}}_k {\bf{h}}_k^{\sf H}{\pmb \Phi}_{\alpha_{\sf bs}}^{1/2} + {\pmb \Phi}_{\beta_{\sf bs}}{\rm diag}\left(\bh_k\bh_k^{\sf H}\right)$. 
Then leveraging the fact that $\|\bar \bv\|^2=1$, we cast the SE $\bar{R}_k$ in \eqref{eq:ee_Rk}  into a Rayleigh quotient form as 
\begin{align}
    \label{eq:se_reform_ee}
    \bar{R}_k = \log_2\left(\frac{\bar{\bv}^{\sf H}\bA_k\bar{\bv}}{\bar{\bv}^{\sf H}\bB_k\bar{\bv}}\right)
\end{align}
where
\begin{align}
    &{\bf{A}}_k = {\rm blkdiag} \big(\bG_k, \cdots, \bG_k\big) + \frac{ \sigma^2}{\tau P} {\bf{I}}_{NK} , \\
    &{\bf{B}}_{k} = {\bf{A}}_{k} - {\rm blkdiag} \left({\bf{0}}_{N\times N}, \cdots, \alpha_k({\pmb \Phi}_{\alpha_{\sf bs}}^{1/2})^{\sf H}{\bf{h}}_k {\bf{h}}_k^{\sf H}{\pmb \Phi}_{\alpha_{\sf bs}}^{1/2} , \cdots, {\bf 0}_{N\times N} \right).
\end{align}
Consequently, we have the following problem for given $\tau$ and $\mu$ 
\begin{align}
    \label{eq:problem_ee_v}
    \mathop{{\text{maximize}}}_{\bar{\bf{v}}}& \;\; \sum_{k = 1}^{K} \log_2\left(\frac{\bar{\bv}^{\sf H}\bA_k\bar{\bv}}{\bar{\bv}^{\sf H}\bB_k\bar{\bv}}\right)  - \mu \sum_{n = 1}^{N} \big(2\,P_{\sf DAC}(b_{{\sf DAC},n},f_s) + P_{\sf RF}\big)  \mathbbm{1}_{\{n \in \CMcal{A}\}}. 
\end{align}
Note that the constraint in \eqref{eq:const_ee_dinkelbach2} is ignored at this step.
In the proposed algorithm, however, $\bar \bv$ will be normalized, which indeed satisfies the constraint. 

Now, the major challenge in solving the problem in \eqref{eq:problem_ee_v} is to handle the indicator function $\mathbbm{1}_{\{n \in \CMcal{A}\}}$.
It is, however, highly difficult due to the lack of smoothness. 
To resolve this challenge, we transform the indicator function into a more favorable form as follows:
let $\tilde {\bf{w}}_{n}$ be the $n$th row of the precoder $\bW$. Then we clarify that  antenna $n$ is active if and only if $\| \tilde {\bf{w}}_{n} \|^2 > 0$.
Equivalently, we consider $\mathbbm{1}_{\{n \in \CMcal{A}\}} = \mathbbm{1}_{\{\|\frac{1}{\sqrt{\alpha_{{\sf bs},n}}}\tilde {\bf{v}}_n\|^2 > 0\}}$ because $\bW \propto {\pmb \Phi}_{\alpha_{\sf bs}}^{-1/2}\bV$ from the definition in \eqref{eq:v_def},  where $\tilde {\bf{v}}_n$ is the $n$th row vector of $\bV$. 
We remark that we omit the effect of $\tau$ since it applies identically to all the antennas, but the quantization loss term is included so that the quantization effect can be applied differently across antennas. Subsequently, we approximate the indicator function by using the following approximation \cite{Sriperumbudur2010AMA}:
\begin{align} \label{eq:indic_approx}
   \mathbbm{1}_{\{|x|^2>0\}} \approx \frac{\log_2 (1 + |x|^2 / \rho)}{\log_2 (1 + 1/\rho)},
\end{align} 
where the approximation becomes tight as $\rho \rightarrow 0$. Using \eqref{eq:indic_approx}, the indicator function in \eqref{eq:P_cir} is represented as 
\begin{align}
    \mathbbm{1}_{\{n \in \CMcal{A}\}} = \mathbbm{1}_{\{\|\tilde {\bf{v}}_n/{\sqrt{\alpha_{{\sf bs},n}}} \|^2 > 0\}} \approx \log_2 \left(1 + {\rho}^{-1} \left\| \frac{1}{\sqrt{\alpha_{{\sf bs},n}}}\tilde {\bf{v}}_n \right\|^2 \right)^{\omega_{\rho}},
\end{align}
where $\omega_{\rho} = 1/\log_2(1 + \rho^{-1})$ and $\rho>0$ is a small enough value. Further denoting that $P_{{\sf ant},n} = 2\,P_{\sf DAC}(b_{{\sf DAC},n},f_s) + P_{\sf RF}$, we have
\begin{align}
    \label{eq:indicator_approx}
     \sum_{n = 1}^{N}  \big(2\,P_{\sf DAC}(b_{{\sf DAC},n},f_s) + P_{\sf RF}) \mathbbm{1}_{\{n \in \CMcal{A}\}}\approx  \sum_{n = 1}^{N} \log_2 \left(1 + \rho^{-1} 
     \left\|\frac{1}{\sqrt{\alpha_{{\rm bs},n}}}\tilde {\bf{v}}_n \right\|^2  \right)^{\omega_{\rho} P_{{\sf ant},n} }.
\end{align}
The next step is to express \eqref{eq:indicator_approx} in terms of the $\bar {\bf{v}}$. To this end, we let $\be_{n}$ be the $N$ dimensional  $n$th canonical basis vector with a single $1$ at its $n$th coordinate and zeros elsewhere. 
Then we can write $\frac{1}{\sqrt{\alpha_{{\rm bs},n}}}\tilde\bv_n$ as $\tilde \bv_n = \be_n^{\sf H}{\pmb \Phi}_{\alpha_{\sf bs}}^{-1/2}\bV$.
Subsequently, we rewrite $\|\frac{1}{\sqrt{\alpha_{{\rm bs},n}}}\tilde \bv_n\|^2$ as
\begin{align}
   \left\|\frac{1}{\sqrt{\alpha_{{\rm bs},n}}}\tilde \bv_n\right\|^2 &= \be_n^{\sf H}{\pmb \Phi}_{\alpha_{\sf bs}}^{-1/2}\bV\bV^{\sf H}{\pmb \Phi}_{\alpha_{\sf bs}}^{-1/2}\be_n
   \\
   & = {\rm vec}(\tilde\be_n^{\sf H}\bV\bV^{\sf H}\tilde\be_n)
   \\
   &\stackrel{(a)} = \left(\left(\tilde\be_n^{\sf T}\bV^*\right) \otimes \tilde\be_n^{\sf H}\right){\rm vec}(\bV)
   \\
   & \stackrel{(b)}= \left(\left(\left(\bI_K \otimes \tilde\be_n^{\sf T}\right){\rm vec}(\bV^*)\right)^{\sf T} \otimes \tilde\be_n^{\sf H}\right){\rm vec}(\bV)
   \\\label{eq:v_row_reform}
   & \stackrel{(c)}={\bar\bv}^{\sf H}\left(\bI_K \otimes \tilde\be_n \otimes \tilde\be_n^{\sf H}\right){\bar\bv},
\end{align}
where $\tilde\be_n = {\pmb \Phi}_{\alpha_{\sf bs}}^{-1/2}\be_n$, $(a)$ and $(b)$ follow from ${\rm vec}(\bA\bB\bC) = (\bC^{\sf T}\otimes \bA){\rm vec}(\bB)$, and $(c)$ comes from $(\bA\otimes\bB)(\bC\otimes\bD)= (\bA\bC)\otimes(\bB\bD)$ and  the definition of $\bar \bv = {\rm vec}(\bV)$.
Using \eqref{eq:indicator_approx} and \eqref{eq:v_row_reform}, 
the objective function in \eqref{eq:problem_ee_v} can be represented as 
\begin{align}
    \label{eq:problem_ee_object_compact}
    L_{v}(\bar\bv) = \log_2 \lambda_{v}(\bar \bv),
\end{align}
where 
\begin{align}
    \lambda_{v}(\bar \bv)  = \prod_{k=1}^{K}\left(\frac{{\bar \bv}^{\sf H} \bA_k\bar \bv}{{\bar \bv}^{\sf H} \bB_k\bar \bv}\right) \prod_{n = 1}^{N} \left(\bar {\bf{v}}^{\sf H} {\bf{E}}_n \bar {\bf{v}} \right)^{-\mu\omega_{\rho} P_{{\sf ant},n}}
\end{align}
and $ {\bf{E}}_n  =  \bI_{NK} + \rho^{-1}\bI_K \otimes \tilde\be_n \otimes \tilde\be_n^{\sf H}$ as $\|\bar \bv\|^2 = 1$.
With the reformulated objective function in \eqref{eq:problem_ee_object_compact}, we derive a condition for a local-optimal stationary point and propose an algorithm to find such a local optimal point.
\begin{lemma}
	\label{lem:kkt_ee}
	 The first-order optimality condition of the optimization problem \eqref{eq:problem_ee_v} with the approximated objective function in \eqref{eq:problem_ee_object_compact} for given $\tau$ and $\mu$ is satisfied if the following holds:
	\begin{align}
		\label{eq:kkt_ee}
	    {\bf{A}}_{\sf KKT} (\bar{\bf{ v}}) \bar {\bf{v}}  = \lambda_{v}(\bar {\bf{v}}) {\bf{B}}_{\sf KKT} (\bar {\bf{v}}) \bar {\bf{v}}, 
	\end{align}
	where
	\begin{align}
		\label{eq:Akkt}
		&{\bf{A}}_{\sf KKT} (\bar {\bf{v}}) = \sum_{k = 1}^{K} \frac{{\bf{A}}_k }{ \bar {\bf{v}}^{\sf H} {\bf{A}}_k \bar {\bf{v}}}\prod_{\ell=1}^{K}\left(\bar{\bv}^{\sf H}\bA_\ell\bar{\bv}\right),
		\\\label{eq:Bkkt}
		&{\bf{B}}_{\sf KKT} (\bar {\bf{v}}) =  \left(\sum_{k = 1}^{K} \frac{{\bf{B}}_k}{\bar {\bf{v}}^{\sf H} {\bf{B}}_k \bar {\bf{v}}}+ \mu\omega_\rho\sum_{n=1}^N P_{{\sf ant},n} \frac{ \bE_n}{\bar\bv^{\sf H} \bE_n \bar\bv}\right) \prod_{\ell=1}^{K}\left(\bar{\bv}^{\sf H}\bB_\ell\bar{\bv}\right)\prod_{m=1}^{N}\left(\bar{\bv}^{\sf H}\bE_m\bar{\bv}\right)^{\mu\omega_\rho P_{{\sf ant},m}}.
	\end{align}
\end{lemma}
\begin{proof}
We compute the first derivative of  \eqref{eq:problem_ee_object_compact} as
\begin{align}
    \nonumber
	\frac{\partial L_{v}(\bar{\bf{v}})}{\partial \bar{\bf{v}}^{\sf H}} = \frac{1}{\lambda_{v}(\bar\bv)\ln 2}\frac{\partial\lambda_{v}({\bar\bv})}{\partial \bar\bv^{\sf H}}
\end{align}
and subsequently, we derive ${\partial\lambda_{v}({\bar\bv})}/{\partial \bar\bv^{\sf H}}$ and set it to zero
\begin{align}
    \label{eq:d_lambda_ee}
	\frac{\partial \lambda_{v}(\bar{\bv})} {\partial \bar{\bv}^{\sf H}} &= 2 \lambda_{v}(\bar\bv)\left(\sum_{k = 1}^{K}   \left( \frac{{\bf{A}}_k\bar {\bf{v}} }{ \bar {\bf{v}}^{\sf H} {\bf{A}}_k \bar {\bf{v}} }  - \frac{{\bf{B}}_k\bar {\bf{v}} }{ \bar {\bf{v}}^{\sf H} {\bf{B}}_k \bar {\bf{v}} } \right) - \mu\omega_\rho\sum_{n=1}^N  P_{{\sf ant},n} \frac{\bE_n\bar\bv}{\bar \bv^{\sf H}\bE_n\bar \bv}\right)
	\\ 
	&= 0.
\end{align}
Then, the first-order optimality condition can be derived as
\begin{align}
	\sum_{k = 1}^{K} \frac{{\bf{A}}_k }{ \bar {\bf{v}}^{\sf H} {\bf{A}}_k \bar {\bf{v}}}\lambda_{v,\sf num}(\bar \bv) \bar {\bf{v}} = \lambda_{v}(\bar {\bf{v}})\left(\sum_{k = 1}^{K} \frac{{\bf{B}}_k }{ \bar {\bf{v}}^{\sf H} {\bf{B}}_k \bar {\bf{v}} }+\mu\omega_\rho\sum_{n=1}^N P_{{\sf ant},n} \frac{ \bE_n}{\bar\bv^{\sf H} \bE_n \bar\bv}\right)\lambda_{v,\sf denom}(\bar \bv)\bar{\bf{v}},
\end{align}
where $\lambda_{v,\sf num}(\bar \bv) =\prod_{\ell=1}^{K}\left(\bar{\bv}^{\sf H}\bA_\ell\bar{\bv}\right)$ and $\lambda_{v,\sf denom}(\bar \bv) = \prod_{\ell=1}^{K}\left(\bar{\bv}^{\sf H}\bB_\ell\bar{\bv}\right)\prod_{m=1}^{N}\left(\bar{\bv}^{\sf H}\bE_m\bar{\bv}\right)^{\mu\omega_\rho P_{{\sf ant},m}}$.
This completes the proof. 
\end{proof}

We interpret the derived condition \eqref{eq:kkt_ee} as a functional eigenvalue problem.

 \begin{remark} \normalfont
 	\label{rm:eig_problem}
 	The derived first-order optimality condition in \eqref{eq:kkt_ee} can be transformed to a functional eigenvalue problem regarding the matrix ${\bf{B}}^{-1}_{\sf KKT}(\bar {\bf{v}}){\bf{A}}_{\sf KKT}(\bar {\bf{v}})$ as
	\begin{align}
		\label{eq:eig_problem}
		c\bar\bv = \lambda(\bar\bv)\bar\bv.
	\end{align}
\end{remark}
From Remark~\ref{rm:eig_problem}, the first-order optimality condition in \eqref{eq:eig_problem} is cast as a functional eigenvalue problem. 
More specifically,  \eqref{eq:eig_problem} is  included  in  a  class  of  nonlinear  eigenvector  dependent  eigenvalue  problem  (NEPv) \cite{cai:siam:18}, where the matrix itself is a function of eigenvectors.
Consequently, treating $\bar {\bf{v}}$ as an eigenvector of ${\bf{B}}_{\sf KKT}^{-1}(\bar {\bf{v}}) {\bf{A}}_{\sf KKT}(\bar {\bf{v}})$,  $\lambda(\bar {\bf{v}})$ is interpreted as a corresponding eigenvalue of ${\bf{B}}_{\sf KKT}^{-1}(\bar {\bf{v}}) {\bf{A}}_{\sf KKT}(\bar {\bf{v}})$. 
In this regard, we need to find the principal eigenvector of ${\bf{B}}^{-1}_{\sf KKT}(\bar {\bf{v}}) {\bf{A}}_{\sf KKT}(\bar {\bf{v}})$ that leads $\lambda(\bar\bv)$ to be a maximum eigenvalue of ${\bf{B}}^{-1}_{\sf KKT}(\bar {\bf{v}}) {\bf{A}}_{\sf KKT}(\bar {\bf{v}})$ and also satisfies the first-order optimality condition in \eqref{eq:kkt_ee}. 
To find such $\bar{\bv}$, we propose a quantization-aware generalized power iteration-based direction optimization algorithm (Q-GPI-DO). 


\begin{algorithm} [t]
\caption{Q-GPI-DO} \label{alg:se} 
{\bf{initialize}}: $\bar {\bf{v}}^{(0)}$\\
Set the iteration count $t = 1$\\
\While {$\left\|\bar {\bf{v}}^{(t)} - \bar {\bf{v}}^{(t-1)} \right\| > \epsilon$ \& $t \leq t_{\max} $}{
Build matrix $ {\bf{A}}_{\sf KKT} (\bar {\bf{v}}^{(t-1)})$ in \eqref{eq:Akkt}\\
Build matrix ${\bf{B}}_{\sf KKT} (\bar {\bf{v}}^{(t-1)})$ in \eqref{eq:Bkkt} \\
Compute $\bar {\bf{v}}^{(t)} = {\bf{B}}^{-1}_{\sf KKT} (\bar {\bf{v}}^{(t-1)}) {\bf{A}}_{\sf KKT} (\bar {\bf{v}}^{(t-1)}) \bar {\bf{v}}^{(t-1)}$ \\
Normalize $\bar {\bf{v}}^{(t)} = \bar {\bf{v}}^{(t)}/\left\| \bar {\bf{v}}^{(t)}\right\|$\\
 $t \leftarrow t+1$}
$\bar {\bf{v}}^\star \leftarrow \bar {\bf{v}}^{(t)}$\\
\Return{\ }{$\bar \bv^\star $}.
\end{algorithm}

Algorithm~\ref{alg:se} describes the proposed algorithm. 
The key idea used in Q-GPI-DO is to modify a power iteration method to be appliable for the corresponding functional eigenvalue problem \eqref{eq:eig_problem}. Specifically, Algorithm~\ref{alg:se} first initializes the stacked precoding vector $\bar \bv^{(0)}$.
Then, the precoding vector $\bar \bv^{(0)}$ is updated  iteratively: at iteration $t$,  the matrices ${\bf{A}}_{\sf KKT} (\bar {\bf{v}}^{(t-1)})$ and ${\bf{B}}_{\sf KKT} (\bar {\bf{v}}^{(t-1)})$ are computed according to \eqref{eq:Akkt} and \eqref{eq:Bkkt}. 
Subsequently, the precoder $\bar {\bf{v}}^{(t)}$ is re-computed as $\bar {\bf{v}}^{(t)} = {\bf{B}}^{-1} _{\sf KKT} (\bar {\bf{v}}^{(t-1)}){\bf{A}}_{\sf KKT} (\bar {\bf{v}}^{(t-1)}) \bar {\bf{v}}^{(t-1)}$ and normalized as $\bar {\bf{v}}^{(t)} = \bar {\bf{v}}^{(t)}/ \left\| \bar {\bf{v}}^{(t)} \right\|$. 
The iteration stops when one of the stopping criteria are met: either converges, i.e., $\| \bar {\bf{v}}^{(t) }- \bar {\bf{v}}^{(t-1)}\| \le \epsilon$ where  $\epsilon > 0 $ denotes a tolerance level or reaches a maximum iteration count $t_{\rm max}$ which may differ depending on a system requirement. 

{\textcolor{black}{
One remarkable benefit of the proposed Q-GPI-DO is that it is not required to exploit any off-the-shelf optimization solver such as CVX. 
Distinguished from other convex-relaxation based approaches, the only computational load of the proposed method is caused when calculating $\bar {\bf{v}}^{(t)} = {\bf{B}}^{-1} _{\sf KKT} (\bar {\bf{v}}^{(t-1)}){\bf{A}}_{\sf KKT} (\bar {\bf{v}}^{(t-1)}) \bar {\bf{v}}^{(t-1)}$. 
As implementing CVX in a practical hardware is nearly infeasible due to the high complexity, Q-GPI-DO has a substantial advantage from a implementation perspective.  
}}




\subsubsection{Optimal Power Scaling $\tau^\star$}
the problem in \eqref{eq:problem_ee_dinkelbach} for given  $\bV$ and $\mu$ reduces to
\begin{align}
    \label{eq:problem_ee_tau}
    \mathop{{\text{maximize}}}_{\tau}& \;\; \sum_{k=1}^K \bar{R}_k -\frac{\mu P}{\kappa} \tau
    \\
    {\text{subject to}}  & \;\;  0<\tau\leq 1 
    \\
    & \;\;  \mu > 0.
\end{align}
The objective function in \eqref{eq:problem_ee_tau} is now concave  with respect to $\tau$. 
In this regard, the optimal $\tau$ can be derived by using a gradient descent algorithm for given $\bar \bv$ and $\mu$.
Let the objective function in \eqref{eq:problem_ee_tau} be $ L_{\tau}(\tau,\bar \bv)$. 
Then the gradient update is given as
\begin{align}
    \label{eq:gd_update}
    \tau^{(t+1)} = \tau^{(t)} + \delta_{\sf GD}\frac{\partial L_{\tau}(\tau^{(t)},\bar\bv)}{\partial \tau^{(t)}},
\end{align}
where $\delta_{\sf GD}$ denotes a step size and  the partial derivative  $\partial L_{\tau}(\tau^{(t)},\bar \bv)/\partial \tau^{(t)}$ is computed as
\begin{align}
    \label{eq:gradient}
    \frac{\partial  L_\tau(\tau,\bar \bv)}{\partial \tau}=\frac{1}{\ln 2}\sum_{k=1}^K \left(\frac{\Xi_k+\Psi_k}{\sigma^2/P+(\Xi_k+\Psi_k)\tau}-\frac{\Psi_k}{\sigma^2/P+\Psi_k\tau}\right) - \frac{\mu P}{\kappa},
\end{align}
where
\begin{align}
    &\Xi_k = \alpha_k|\bh^{\sf H}_k {\pmb \Phi}_{\alpha_{\sf bs}}^{1/2}\bv_k|^2, 
    \\
    &\Psi_k = \sum_{\ell = 1}^{K}  |\bh^{\sf H}_k{\pmb \Phi}_{\alpha_{\sf bs}}^{1/2} \bv_\ell|^2 - \alpha_k|\bh_k^{\sf H}{\pmb \Phi}_{\alpha_{\sf bs}}^{1/2}\bv_k|^2  +  \sum_{\ell = 1}^{K}\bv_\ell^{\sf H}{\pmb \Phi}_{\beta_{\sf bs}}{\rm diag}\left(\bh_k\bh_k^{\sf H}\right)\bv_\ell.
\end{align}
To decide $\delta_{\sf GD}$, we use a backtracking line search method \cite{} in simulations. 

Recall that we have the constraint $0<\tau \leq 1$. 
From \eqref{eq:gd_update}, however, it is not guaranteed to satisfy the constraint. 
Accordingly, for each update, we perform thresholding of $\tau$; if $\tau>1$, then set $\tau= 1$, and if $\tau<0$, then set $\tau = 0$.

\subsubsection{Dinkelbach Update $\mu$}

Once the precoder $\bV$ and scaling factor $\tau$ are derived, we update $\mu$ for next outer iteration as
\begin{align}
    \label{eq:mu_update}
    \mu =\frac{\sum_{k = 1}^{K} \bar{R}_k\left(\bV,\tau\right)}{P_{\sf LO}  + \sum_{m = 1}^{N}\log_2 \left(1 + {\rho}^{-1} \left\| \frac{1}{\sqrt{\alpha_{{\sf bs},m}}}\tilde {\bf{v}}_m \right\|^2 \right)^{\omega_{\rho}}\big(2\,P_{\sf DAC}(b_{{\sf DAC},m},f_s) + P_{\sf RF}\big) + {\tau P}/{\kappa}}.
\end{align}
\begin{algorithm} [t]
\caption{Q-GPI-EEM} \label{alg:ee} 
{\bf{initialize}}: $\bar {\bf{v}}^{(0)}$, $\tau^{(0)}$, and $\mu^{(0)}$\\ 
Set the iteration count $n_0= 1$\\
\While {$|\mu^{(t_0)} - \mu^{(t_0-1)}|/|\mu^{(t_0)}| > \epsilon_0 $ \& $t_0 \leq t_{\max} $}{
Set the iteration count $t_1 = 1$\\
\While {$\left\| {\bf{W}}^{(t_1)} - {\bf{W}}^{(t_1-1)} \right\|_{F} /\left\| {\bf{W}}^{(t_1)}\right\|_{F}> \epsilon_1$ \& $t_1 \leq t_{\max} $}{
Set the iteration count $n_2 = 1$\\
\While {$|\tau^{(t_2)} - \tau^{(t_2-1)}|/|\tau^{(t_2)}| > \epsilon_2$ \& $t_2 \leq t_{\max} $}{
Set $\partial L_{\tau}(\tau^{(t_2)}, \bar{\bf{v}})/\partial \tau^{(t_2)}$ according to \eqref{eq:gradient}\\
Update $\tau^{(t_2)}$ according to  \eqref{eq:gd_update}\\
Thresholding $\tau$\\
$t_2 \leftarrow t_2 + 1$
}
$\bar {\bf v} = \text{Q-GPI-DO}(\bA_{\sf KKT}(\bar{\bv}, \tau^{(t_2)}), \bB_{\sf KKT}(\bar{\bv},\tau^{(t_2)}))$\\
Compute ${\bW}^{(t_1)} = \sqrt{\tau^{(t_2)}}{\pmb \Phi}_{\alpha_{\sf bs}}^{-1/2} \left[{\bf{v}}_1, {\bf{v}}_2, \dots, {\bf{v}}_K\right] $\\
 $t_1 \leftarrow t_1+1$
}
Update $\mu^{(t_0)}$ according to \eqref{eq:mu_update}\\
 $t_0 \leftarrow t_0+1$
}
Set $\tilde \bw_n^{(t_1)} = {\bf 0}_{1\times K}$ if $\|\frac{1}{\sqrt{\alpha_{{\rm bs},n}}}\hat \bv_{n}\|^2< \epsilon_{\sf as}$, $ n = 1,\dots, N$\\
\Return{\ }{${\bW}^{(t_1)}$}
\end{algorithm}

Algorithm~\ref{alg:ee} describes the proposed quantization-aware GPI-based EE maximization algorithm (Q-GPI-EEM).
With initialization, the algorithm computes $\tau$ by using the gradient descent method for given $\bar \bv$ and $\mu$ and also finds $\bar \bv$ by using  the Q-GPI-DO algorithm for given $\tau$ and $\mu$. 
Then $\bW$ is computed from the derived $\tau$ and $\bar \bv$. 
The algorithm repeats these steps until $\bW$ converges. 
Once converged, the outer loop computes $\mu$ by using the derived $\bW$ and repeats the previous steps until $\mu$ converges. 
Once $\mu$ converged, we check whether  $\|\frac{1}{\sqrt{\alpha_{{\rm bs},n}}}\tilde \bv_{n}\|^2\geq \epsilon_{\sf as}$ to select antennas which have  effective gain where $\epsilon_{\sf as} > 0$ is a small enough value. 
This selection approach is effective as the rows of $\bV$ are jointly designed with each other and thus, the designed $\bV$ incorporates relative gains across the antennas.
The norm of $\tilde \bv_n$, however, highly depends on the number of antennas  because the norm of $\bar \bv$ is limited as $\|\bar \bv\| = 1$.
In this regard, the threshold $\epsilon$ may vary with $N$.
To avoid this issue, we normalize $\tilde \bv_n$ by  ${\max_i \|\frac{1}{\sqrt{\alpha_{{\rm bs},i}}}\tilde \bv_i\|}$, i.e., we perform 
\begin{align}
    \hat \bv_n = \tilde \bv_n /{\max_i \left\|\frac{1}{\sqrt{\alpha_{{\rm bs},i}}}\tilde \bv_i\right\|}
\end{align}
and set $\tilde \bw_n = {\bf 0}_{1\times K}$ if $\|\frac{1}{\sqrt{\alpha_{{\rm bs},n}}}\hat \bv_{n}\|^2< \epsilon_{\sf as}$ $\forall n$, where $\tilde \bw_n$ represents the $n$th row of $\bW$.
Finally, the proposed method returns $\bW$.

\begin{remark}[Complexity] \normalfont
The computational complexity for each step in the proposed Q-GPI-EEM is analyzed as follows. 
In the Q-GPI-DO step, the main computational load is caused by calculating the inversion of ${\bf{B}}_{\sf KKT}(\bar {\bf{f}})$. Since ${\bf{B}}_{\sf KKT}(\bar {\bf{f}})$ is comprised of $NK \times NK$ block-diagonal matrices, the inversion can be obtained by calculating the inversion of each sub-matrix. This results in that the complexity of Q-GPI-EEM per iteration is characterized as $\cO (KN^3)$. In the precoding power scaling optimization step, we only find a scalar power scaling value, its computational complexity per iteration is $\cO(1)$. 
As a result, the total computational complexity of the proposed Q-GPI-EEM is an order of $\cO (KN^3)$. 

\end{remark}

{\textcolor{black}{
\begin{remark}[Convergence] \normalfont
It is challenging to prove the convergence of the proposed method rigorously. The main obstacle is Q-GPI-DO. 
As explained above, through the lens of a functional eigenvalue problem, Q-GPI-DO is interpreted as a solution of a NEPv \cite{cai:siam:18}. 
Conventionally, in the NEPv, the self-consistent field iteration (SCF) can be used to solve the problem. The SCF iteratively performs the eigenvector decomposition based on the previously obtained eigenvectors. This is similar to our Q-GPI-DO in the sense that Q-GPI-DO iteratively performs the power iteration based on the previously obtained the principal eigenvector. 
Considering the relationship between the eigenvector decomposition and the power iteration, the SCF and Q-GPI-DO are closely related. For instance, in a typical eigenvalue problem, it is well known that the power iteration can track the results of the eigenvector decomposition. 
Extending this notion to the NEPv is a key for the convergence proof of the proposed method. We leave this as future work. 
In the later section, we empirically show that our method converges very well in general system environments. 
\end{remark}
}}

\section{Special Case: Sum Spectral Efficiency Maximization}

We can show that a sum SE maximization problem in the considered massive MIMO system with low-resolution DACs and ADCs is regarded as a special case of the EE maximization.
Accordingly, we provide a brief description of finding precoders that maximize the sum SE by exploiting the results derived in the EE.
The sum SE maximization problem is formulated by setting $\mu$ in \eqref{eq:problem_ee_dinkelbach} to zero as
\begin{align} 
    \label{eq:problem_se}
    \mathop{{\text{maximize}}}_{\bW}& \;\; \sum_{k = 1}^{K} R_k \\
    \label{eq:const_se}
    {\text{subject to}} & \;\;  {\rm tr}\left(\bbE\big[\bx_{\sf q}\bx_{\sf q}^{\sf H}\big]\right) \leq P.
\end{align}
Since the SINR in \eqref{eq:sinr_reform} increases with the transmit power, 
the SE is maximized when we meet
\begin{align}
    \label{eq:txpower_se}
    {\rm tr}\left({\pmb \Phi}_{\alpha_{\sf bs}}\bW\bW^{\sf H}\right) = 1,
\end{align}
i.e., the BS transmits signals use the maximum transmit power $P$.
Since this is equivalent to having $\tau = 1$, the results derived in the EE problem can be directly used by setting $\tau = 1$ and $\mu = 0$.
Now let us define a weighted precoding vector as
\begin{align}
    \label{eq:f_def}
    \bff_k = {\pmb \Phi}_{\alpha_{\sf bs}}^{{1}/{2}}\bw_k
\end{align} 
which is equivalent to $\bv_k$ in \eqref{eq:v_def} with $\tau = 1$.
Then we have the following optimality condition:
\begin{lemma}
	\label{lem:kkt}
	 The first-order optimality condition of the optimization problem \eqref{eq:problem_se} is satisfied if the following holds:
	\begin{align}
		\label{eq:kkt}
	    {\bf{C}}_{\sf KKT} (\bf{\bar f}) \bar {\bf{f}}  = \lambda(\bar {\bf{f}}) {\bf{D}}_{\sf KKT} (\bar {\bf{f}}) \bar {\bf{f}}, 
	\end{align}
	where
	\begin{align}
			\label{eq:Ckkt}
		&{\bf{C}}_{\sf KKT} (\bar {\bf{f}}) = \sum_{k = 1}^{K} \frac{{\bf{C}}_k }{ \bar {\bf{f}}^{\sf H} {\bf{C}}_k \bar {\bf{f}}}\prod_{\ell=1}^{K}\left(\bar{\bff}^{\sf H}\bC_\ell\bar{\bff}\right),\\
		\label{eq:Dkkt}
		&{\bf{D}}_{\sf KKT} (\bar {\bf{f}}) =  \sum_{k = 1}^{K} \frac{{\bf{D}}_k }{ \bar {\bf{f}}^{\sf H} {\bf{D}}_k \bar {\bf{f}} }\prod_{\ell=1}^{K}\left(\bar{\bff}^{\sf H}\bD_\ell\bar{\bff}\right),\\
     &{\bf{C}}_k = {\rm blkdiag} \big(\bG_k, \cdots, \bG_k\big) + \frac{ \sigma^2}{P} {\bf{I}}_{NK}, \\
     &{\bf{D}}_{k} = {\bf{C}}_{k} - {\rm blkdiag} \left({\bf{0}}_{N\times N}, \cdots,  \alpha_k({\pmb \Phi}_{\alpha_{\sf bs}}^{1/2})^{\sf H}{\bf{h}}_k {\bf{h}}_k^{\sf H}{\pmb \Phi}_{\alpha_{\sf bs}}^{1/2} , \cdots, {\bf 0}_{N\times N} \right),\\
     	& \lambda(\bar{\bff}) = \prod_{k=1}^{K}\left(\frac{\bar{\bff}^{\sf H}\bC_k\bar{\bff}}{\bar{\bff}^{\sf H}\bD_k\bar{\bff}}\right).
	\end{align}
	\begin{proof}
	    We derive the results by setting $\mu = 0$ and $\tau =1$ from Lemma~\ref{lem:kkt_ee}. 
	\end{proof}
\end{lemma}

We note that we can solve \eqref{eq:kkt} by leveraging Algorithm~\ref{alg:se} with $\mu = 0$ and $\tau = 1$; replacing \eqref{eq:Akkt}, \eqref{eq:Bkkt}, and $\bv_k$ with \eqref{eq:Ckkt}, \eqref{eq:Dkkt}, and $\bff_k$, respectively and computing $\bW$ based on \eqref{eq:f_def}. The algorithm provides the best precoder that maximizes the sum SE among all the stationary points. 
We call the quantization-aware GPI-based SE maximization algorithm as Q-GPI-SEM.

\section{Numerical Results}

This section evaluates the proposed algorithms to validate the performance and draw key system design insights. We also evaluate benchmark algorithms for comparison. The following cases are included in the simulations: $(1)$ the proposed algorithms, $(2)$ quantization-ignorant conventional GPI-based SE maximization (GPI-SEM) \cite{choi:twc:20}, $(3)$ quantization-aware linear precoders such as regularized zero-forcing (Q-RZF), zero-forcing (Q-ZF), and maximum ratio transmission (Q-MRT), and $(4)$ conventional linear precoders such as RZF, ZF, and MRT. 
The quantization-aware linear precoders are derived based on the AQNM system model.

\subsection{Simulation Environments}

We adopt a one-ring model \cite{adhi:tit:13} to generate the channel vector $\bh_k$. 
To generate pathloss $\rho_k$, we adopt the log-distance pathloss model in \cite{erceg1999empirically}; cell radius is $1$ km, the minimum distance between the BS and users is $100$ m, pathloss exponent is $4$, and $2.4$ GHz carrier frequency with $100$ MHz bandwidth, $8.7$ dB lognormal shadowing variance, and $5$ dB noise figure are considered.

For the BS power consumption, we set $P_{\sf LP} =14$ mW, $P_{\sf M} = 0.3$ mW, $P_{\sf LO} = 22.5$ mW, $P_{\sf H} = 3$ mW, and $\kappa = 0.27$ \cite{ribeiro:jstsp:18}.
For the parameters used in the proposed algorithms, we set $t_{\sf max} = 10$, $\epsilon = 10^{-8}$, $\epsilon_{\sf as} = 0.05$ $\epsilon_0 = 0.001$, and $\epsilon_{\sf gpi}$, $\epsilon_1$, $\epsilon_2$ to be $0.1$ unless mentioned otherwise, and set $\delta_{\sf GD} = 1$ for an initial step size which will be updated according to the backtracking line search algorithm \cite{armijo1966minimization}.

\subsection{Energy Efficiency Evaluation}

 \begin{figure}[!t]
	\centerline{\resizebox{0.6\columnwidth}{!}{\includegraphics{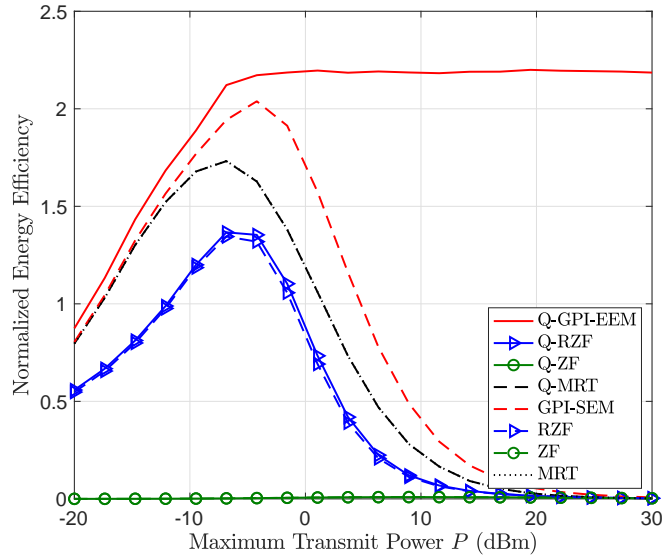}}}     
	\caption{The energy efficiency versus maximum transmit power $P$ results for $N= 32$ BS antennas, $K = 8$ users, $b_{{\sf DAC},n} \sim {\rm Unif}[2, 12]$, and $b_{{\sf ADC},k} \sim {\rm Unif}[2, 6]$.}
 	\label{fig:EePerPower}
\end{figure}

We evaluate the proposed algorithm in terms of the EE performance. 
Fig.~\ref{fig:EePerPower} shows the EE versus maximum transmit power $P$ results for $N= 32$ BS antennas, $K = 8$ users, $b_{{\sf DAC},n} \sim {\rm Unif}[2, 12]$, and $b_{{\sf ADC},k} \sim {\rm Unif}[2, 6]$.
The proposed Q-GPI-EEM algorithm achieves the highest EE performance while maintaining the highest EE for any maximum transmit power regime once it attains the highest EE. 
Accordingly, the proposed algorithm is shown to be robust to the maximum available transmit power by maintaining its highest EE and thus, provides high EE performance regardless of the maximum available transmit power.

 \begin{figure}[!t]
	\centerline{\resizebox{0.55\columnwidth}{!}{\includegraphics{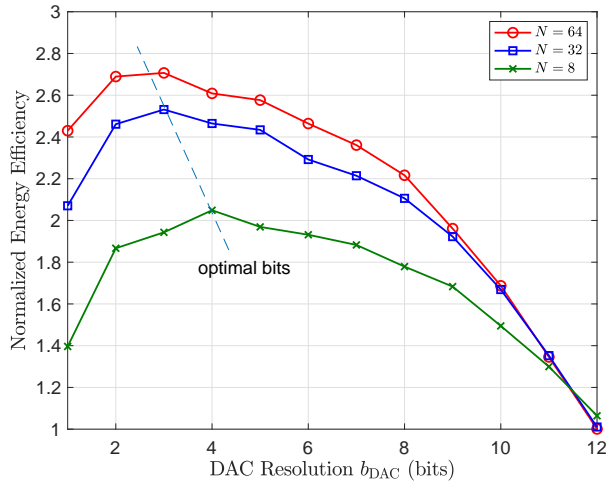}}}     
	\caption{The energy efficiency versus DAC resolutions $b_{\rm DAC}$ for $N\in \{8, 32, 64\}$ BS antennas, $K = 8$ users, $P = 30$ dBm maximum transmit power, and $b_{{\sf ADC},k} = 10$ ADC bits.}
 	\label{fig:EePerDAC}
\end{figure}

The EE with respect to the number of DAC bits of the proposed Q-GPI-EEM is evaluated for $N\in \{8, 32, 64\}$ BS antennas, $K = 8$ users, $P = 30$ dBm maximum transmit power, and $b_{{\sf ADC},k} = 10$ ADC bits in Fig.~\ref{fig:EePerDAC}. 
In this case, the resolution of all DACs is the same, and the selection threshold is relaxed to $\epsilon_{\sf as} = 0.2$ when $N=64$.
We note that the EEs of all the cases can be maximized when 3 to 4 bits are used per DAC. 
In addition, the optimal number of bits tends to become smaller, and the EE increases with the number of antennas.
In this regard, using more antennas with the proposed algorithm provides gains in both the SE and EE, and also it allows the BS to use coarser quantizers, thereby saving more power and simplifying each RF chain.

 \begin{figure}[!t]
	\centerline{\resizebox{0.6\columnwidth}{!}{\includegraphics{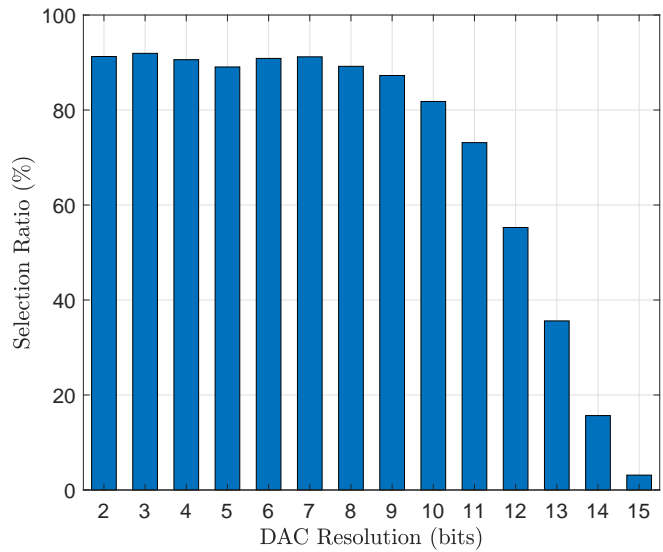}}}     
	\caption{Selection ratio of each DAC resolution for $N= 42$ BS antennas, $K = 8$ users, $b_{{\sf ADC},k} \sim {\rm Unif}[2, 6]$, and $P = 30$ dBm maximum transmit power. Each DAC resolution is assigned to three BS antennas.}
 	\label{fig:EePerDACHist}
\end{figure}

 \begin{figure}[!t]
	\centerline{\resizebox{0.6\columnwidth}{!}{\includegraphics{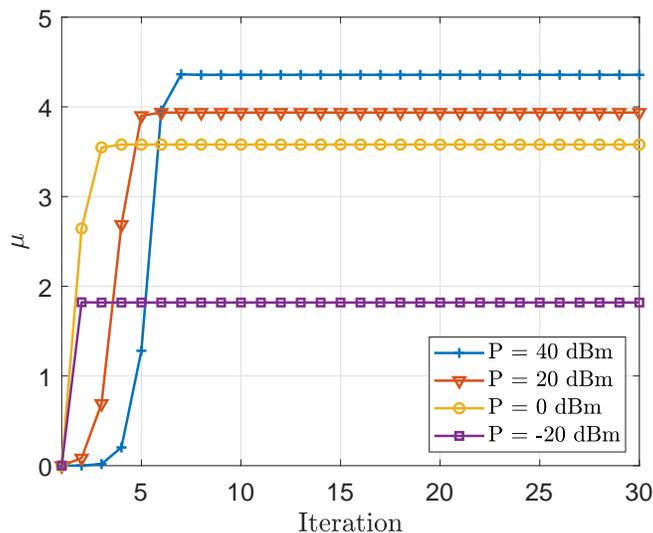}}}     
	\caption{Convergence of $\mu$ for $N= 32$ BS antennas, $K=8$ users, $b_{{\sf DAC},n} \sim {\rm Unif}[2, 12]$,  $b_{{\sf ADC},k} \sim {\rm Unif}[2, 6]$, and $P \in \{-20, 0, 20, 40\}$ dBm maximum transmit power.}
 	\label{fig:MuPerIteration}
\end{figure}

Now, we present the DAC selection ratios for a different number of bits. 
In the simulation, we set $N= 42$ BS antennas, $K = 8$ users, $b_{{\sf ADC},k} \sim {\rm Unif}[2, 6]$, and $P = 30$ dBm maximum transmit power. 
Each DAC resolution is assigned to three BS antennas.
Accordingly, the selection ratio indicates the ratio of selected antennas within each resolution on average.
For example, the selection ratio is about $36\%$ for $13$ bits, which means that only one antenna out of the three antennas with $13$ bits is selected at each transmission on average.
It is shown in Fig.~\ref{fig:EePerDACHist} that about $90\%$ of the antennas are selected for the low-to-medium resolution DACs. 
The selection ratio, however, rapidly decreases with the number of DAC bits in the high-resolution regime since such antennas with high-resolution DACs consume unnecessarily high power, which corresponds to our intuition. 
Consequently, less than one antenna out of the three antennas with $14$ bits is selected at each transmission on average, and no antenna with $15$ bits is selected in most cases.

To exam the convergence of the proposed Q-GPI-EEM algorithm, we provide the convergence of $\mu$ with respect to iteration for $N= 32$ BS antennas, $K=8$ users, $b_{{\sf DAC},n} \sim {\rm Unif}[2, 12]$,  $b_{{\sf ADC},k} \sim {\rm Unif}[2, 6]$, and $P \in \{-20, 0, 20, 40\}$ dBm maximum transmit power.
As we observe in Fig.~\ref{fig:MuPerIteration}, $\mu$ converges fast for all cases.
Although convergence takes longer with higher $P$, the result in Fig.~\ref{fig:MuPerIteration} still shows that $\mu$ converges within $8$ iterations, which can guarantee a fast convergence of Q-GPI-EEM in the practical maximum transmit power regime.

\subsection{Spectral Efficiency Evaluation}

 \begin{figure}[!t]
	\centerline{\resizebox{0.6\columnwidth}{!}{\includegraphics{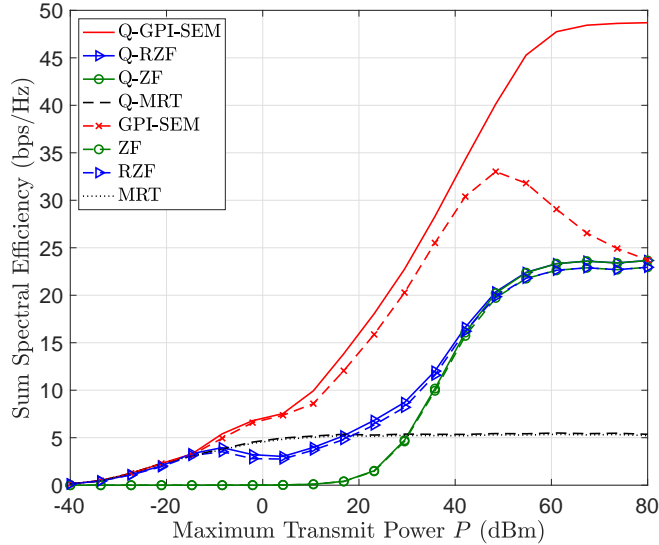}}}     
	\caption{The sum spectral efficiency versus maximum transmit power $P$ results for $N=32$ BS antennas, $K=8$ users, $b_{{\sf DAC},n} \sim {\rm Unif}[2, 12]$, and $b_{{\sf ADC},k} \sim {\rm Unif}[2, 6]$.}
 	\label{fig:SePerPower}
\end{figure}
Fig.~\ref{fig:SePerPower} shows the sum SE versus maximum transmit power $P$ results for $N=32$ BS antennas, $K=8$ users, $b_{{\sf DAC},n} \sim {\rm Unif}[2, 12]$, and $b_{{\sf ADC},k} \sim {\rm Unif}[2, 6]$. 
As shown in Fig.~\ref{fig:SePerPower}, Q-GPI-SEM achieves the highest sum SE over the most range of $P$. 
As the maximum transmit power increases, the SE is saturated because the quantization distortion also increases. 
The proposed algorithm can pull up this saturation level more than $2$ times that of the conventional methods.
We note that GPI-SEM provides a higher sum SE than the RZF, ZF, and MRT-based precoders. 
GPI-SEM, however, shows a huge gap from Q-GPI-SEM as $P$ increases, i.e., quantization noise also increases.
Moreover, its SE even decreases in the very high transmit power regime because the interference from the quantization error, which cannot be fully treated with GPI-SEM dominates the SE performance in the regime.
In the Q-RZF/RZF, Q-ZF/ZF, and Q-MRT/MRT cases, the sum SE performance is not comparable with that of Q-GPI-SEM except in the very low transmit power regime where the quantization error is buried in the thermal noise.
Therefore, Fig.~\ref{fig:SePerPower} validates the sum SE performance of the proposed method over the practical transmit power regime.




\begin{figure}[t]
\centering
$\begin{array}{c c }
{\resizebox{0.5\columnwidth}{!}
{\includegraphics{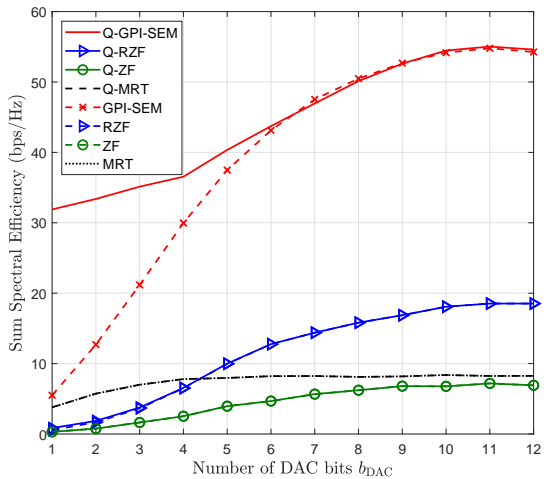}}
} &
{\resizebox{0.52    \columnwidth}{!}
{\includegraphics{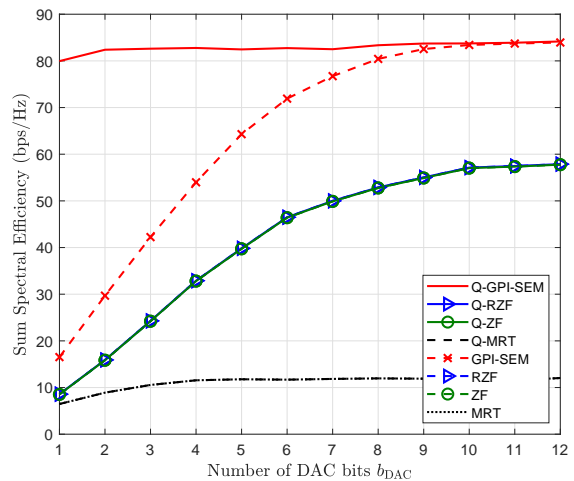}}
}\\ \mbox{\small (a) $N=8$ BS antennas} & \mbox{\small (b) $N=32$ BS antennas}
\end{array}$
\caption{
The sum spectral efficiency versus the number of DAC bits $b_{\rm DAC}$ for $K=8$ users, $b_{{\sf ADC},k} =10$ ADC bits for all $k$, $P=50$ dBm maximum transmit power, and $N\in \{8, 32\}$ BS antennas.}
\label{fig:SePerDAC}
\end{figure}

 \begin{figure}[!t]
	\centerline{\resizebox{0.6\columnwidth}{!}{\includegraphics{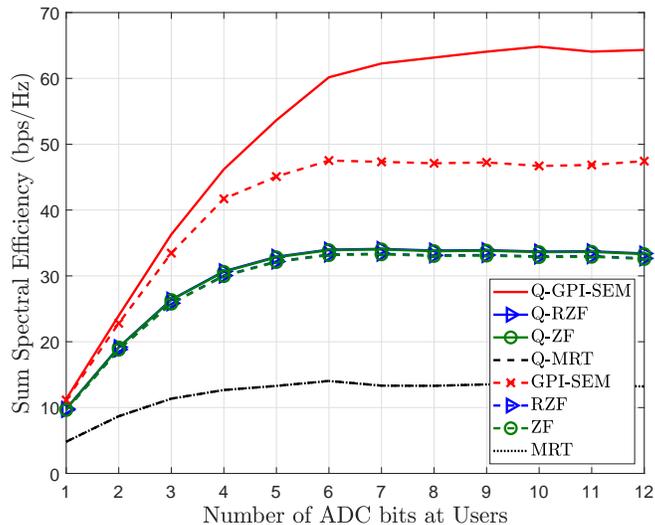}}}     
	\caption{The sum spectral efficiency versus the number of ADC bits $b_{\rm ADC}$ for $N=32$ BS antennas, $K=8$ users, $b_{{\sf DAC},n} \sim {\rm Unif}[2, 12]$, and $P=30$ dBm maximum transmit power.}
 	\label{fig:SePerADC}
\end{figure}

Fig.~\ref{fig:SePerDAC} shows the sum SE versus the number of DAC bits $b_{\rm DAC}$ for $K=8$ users, $b_{{\sf ADC},k} =10$ ADC bits for all $k$, $P=50$ dBm maximum transmit power, and $N\in \{8, 32\}$ BS antennas. 
In this case, the entire DACs have the same resolution. 
We first note that Q-GPI-SEM provides the highest sum SE, and the sum SE of GPI-SEM converges to that of Q-GPI-SEM as the number of DAC bits increases since both the DACs and ADCs have high resolutions. 
More importantly, the performance variation of Q-GPI-SEM over  DAC resolutions becomes marginal in $N=32$ case compared to $N=8$ case, whereas the other algorithms still show high-performance variations.  
Therefore, in the massive MIMO, the proposed algorithm is highly robust to quantization noise at the BS by achieving the SE of high-resolution DACs with low-resolution DACs ($2\sim 5$ bits).

The sum SE with respect to the number of ADC bits $b_{\rm ADC}$ is also simulated for $N=32$ BS antennas, $K=8$ users, $b_{{\sf DAC},n} \sim {\rm Unif}[2, 12]$, and $P=30$ dBm maximum transmit power in Fig.~\ref{fig:SePerADC}. 
In this case, the entire ADCs have the same resolution. 
As expected, the sum SEs of all algorithms increase with $b_{\rm ADC}$ and saturate when the quantization error from ADCs becomes negligible. 
Q-GPI-SEM, however, still achieves the highest sum SE in the high ADC resolution regime, showing a large gap from the other methods since the quantization error from the low-resolution DACs still remains.

Overall, the proposed Q-GPI-SEM algorithm outperforms the conventional precoding algorithms, and it is indeed an efficient method in the massive MIMO communication systems with low-resolution DACs or ADCs, showing robustness to quantization error. 

\begin{figure}[t]
\centering
$\begin{array}{c c }
{\resizebox{0.5\columnwidth}{!}
{\includegraphics{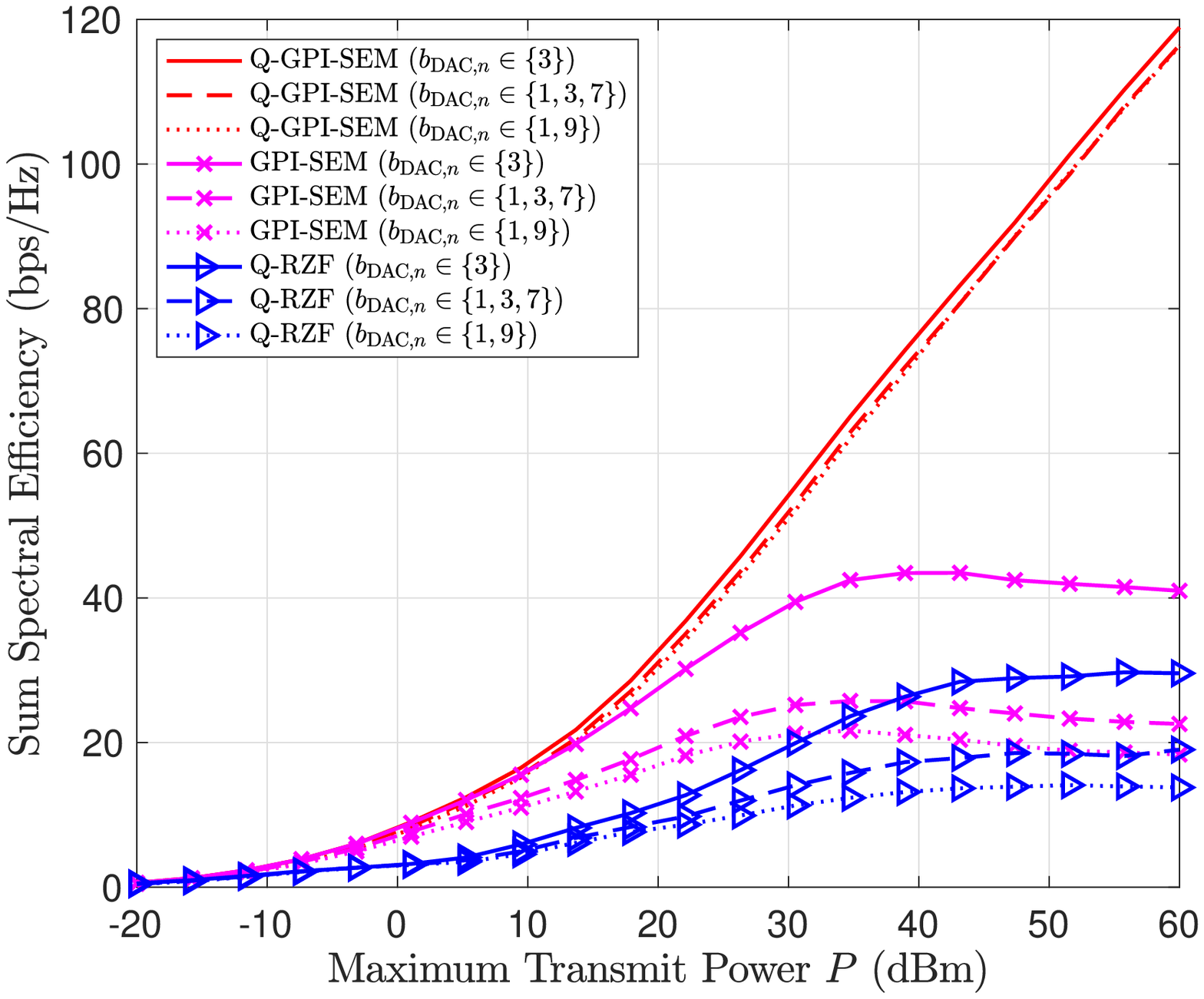}}
} &
{\resizebox{0.5   \columnwidth}{!}
{\includegraphics{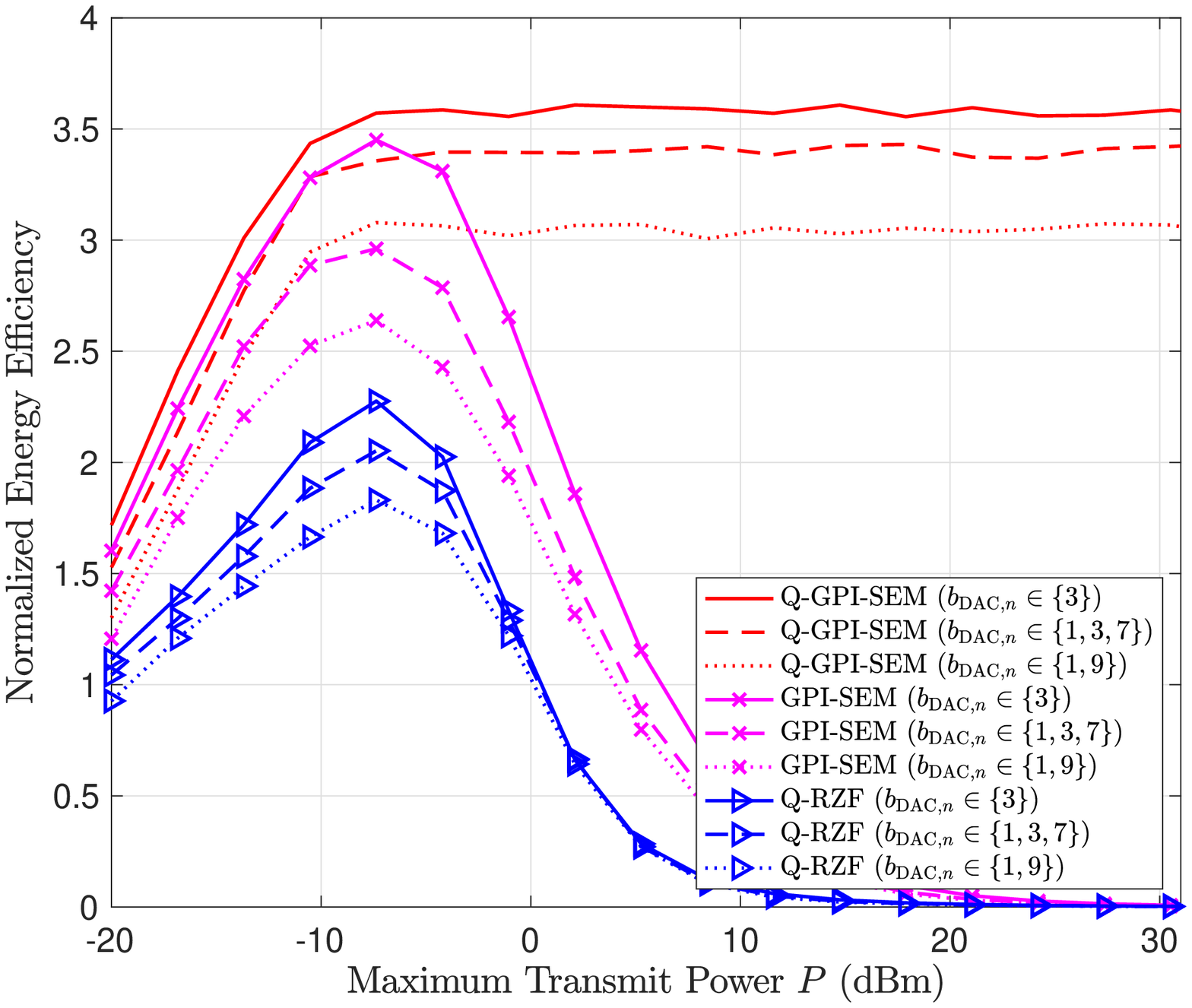}}
}\\ \mbox{\small (a) Spectral Efficiency} & \mbox{\small (b) Energy Efficiency}
\end{array}$
\caption{
The spectral efficiency and energy efficiency versus maximum transmit power $P$ for $K=8$ users, $b_{{\sf ADC},k} =10$ ADC bits for all $k$, and $N = 32$ BS antennas with different DAC configurations.}
\label{fig:EeDACconfig}
\end{figure}

Finally, to provide system design insights, we compare various DAC configurations such as $(i)$ $b_{{\sf DAC},n} = 3$, $(ii)$ $b_{{\sf DAC},n} \in \{1, 3, 7\}$, and $(iii)$ $b_{{\sf DAC},n} \in \{1,9\}$ under the constraint of the total number of DAC resolution bits ($96$ bits total) for $K=8$, $b_{{\sf ADC},k} =10$ for all $k$, and $N = 32$.
Fig.~\ref{fig:EeDACconfig} reveals that the homogeneous DAC configuration where all DACs have the same resolution achieves the highest sum SE and EE. 
In particular, the homogeneous DAC configuration provides a noticeable improvement in the EE since medium- and high-resolution DACs in the other configurations cause inefficiency in power consumption.
In addition, the proposed algorithm shows a relatively small variation in the sum SE across different DAC configurations, whereas the other algorithms reveal a noticeable performance gap across configurations. 
Therefore, the result demonstrates that the proposed algorithm is more robust to DAC configurations, providing more system design flexibility.

\section{Conclusion}

In this paper, we solved a precoding problem for EE maximization in downlink multiuser massive MIMO systems with low-resolution DACs and ADCs. 
To take into account the effects of RF circuit power consumption, 
we incorporated an antenna selection feature into the EE maximization problem.
Managing the quantization errors, we 
reformulated the SINR and adopted the Dinkelbach method. 
Subsequently, we decomposed the problem into precoding direction and power scaling problems and proposed the joint precoding and antenna selection algorithm.
As a special case, we showed that the proposed algorithm can reduce to the SE maximization algorithm by leveraging the product of Rayleigh quotients form of the reformulated SINR.
The simulation results demonstrated that the proposed algorithms improve both EE and SE compared to conventional methods. In particular, the EE maximization algorithm presented robustness to the maximum available transmit power constraint with fast convergence, while other methods suffer from EE degradation as the maximum transmit power increases. 
In addition, it was shown that the proposed methods achieve high enough EE and SE even with low-resolution DACs in the massive MIMO regime, which means that the performance degradation caused by low-resolution quantizers can be compensated by using our method with large-scale arrays. 
As a result, the proposed algorithms can provide considerable benefits in the future massive MIMO systems by offering high flexibility on quantizer configuration and improving the SE and EE performance.

\bibliographystyle{IEEEtran}
\bibliography{ref_SE_EE_QMIMO}

\end{document}